\documentclass[journal,10pt,twoside]{IEEEtranTCOM}
\normalsize
\newif\ifonecolumn

\usepackage{citesort}
\usepackage{enumerate}
\usepackage{psfrag}
\usepackage{subfigure}
\usepackage{url}
\usepackage{stfloats}
\usepackage{amsmath}
\usepackage{array}
\usepackage{amsfonts}
\usepackage{amssymb}
\usepackage{multirow}
\usepackage{bigstrut}
\usepackage{bigdelim}
\usepackage{color}
\usepackage{balance}
\usepackage{float}
\usepackage{comment}
\usepackage[final]{graphicx}

\newtheorem{example}{Example}

\newtheorem{proposition}{Proposition}

\newtheorem{theorem}{Theorem}

\DeclareMathOperator*{\erfc}{erfc}

\newcommand\T{\rule{0pt}{2.5ex}}
\newcommand\B{\rule[-1.1ex]{0pt}{0pt}}

\begin{document}

\title{Near-Field Passive RFID Communication: Channel Model and Code Design}%
\author{\'{A}ngela~I.~Barbero, Eirik~Rosnes,~\IEEEmembership{Senior~Member,~IEEE}, 
Guang~Yang,  and {\O}yvind~Ytrehus,~\IEEEmembership{Senior~Member,~IEEE}%
\thanks{This work was supported by the Research Council of Norway through the ICC:RASC project, by the Spanish Ministerio de Ciencia e Innovaci\'{o}n through project MTM2010-21580-CO2-02, and by Simula@UiB. This work was presented in part at the 2011 Information Theory and Applications (ITA) workshop, San Diego, CA, Feb.\ 2011, and in part at the 3rd International Castle Meeting on Coding Theory and Applications (3ICMTA), Castell de Cardona, Cardona, Spain, Sep.\ 2011. }%
\thanks{\'{A}.\ I.\ Barbero is with the Departamento de Matem\'{a}tica Aplicada, Universidad de Valladolid, 47011 Valladolid, Spain. E-mail: angbar@wmatem.eis.uva.es.}%
\thanks{E.\ Rosnes is with the Selmer Center, Department of Informatics, University of Bergen, N-5020 Bergen, Norway, and the Simula Research Lab. E-mail: eirik@ii.uib.no.}%
\thanks{G.\ Yang was with the Selmer Center, Department of Informatics, University of Bergen, N-5020 Bergen, Norway. She is now with the Norwegian Social Science Data Services (NSD). E-mail: guang.yang@nsd.uib.no.}%
\thanks{{\O}.\ Ytrehus is with the Selmer Center, Department of Informatics, University of Bergen, N-5020 Bergen, Norway, and the Simula Research Lab. E-mail: oyvind@ii.uib.no.}}%
%


\maketitle

\begin{abstract}
This paper discusses a new channel model and code design  for the \emph{reader-to-tag} channel in near-field passive radio frequency identification (RFID) systems using inductive coupling as a power transfer mechanism.
%
If the receiver resynchronizes its internal clock each time a bit is detected, the bit-shift channel used previously in the literature to model the reader-to-tag channel needs to be modified. In particular, we propose a discretized Gaussian shift channel as a new channel model in this scenario. We introduce the concept of quantifiable \emph{error avoidance}, which is much simpler than error correction. The capacity is computed numerically, and we also design some new simple codes for error avoidance on this channel model based on insights gained from the capacity calculations.  Finally, some simulation results are presented to compare the proposed codes to the Manchester code and two previously proposed codes for the bit-shift channel model.
%
%
%
%
\end{abstract}


\section{Introduction}
Inductive coupling is a technique by which energy from one circuit is transferred to another without wires. Simultaneously, the energy transfer can be used as a vehicle for information transmission. This is a fundamental technology for near-field passive radio frequency identification (RFID) applications as well as lightweight sensor applications.

In the passive RFID application, a \emph{reader}, containing or attached to a power source, controls and powers a communication session with a \emph{tag}; a device without a separate power source. The purpose of the communication session may be, for examples, object identification, access control, or acquisition of sensor data.

Several standards exist that specify lower layer coding for RFID protocols. However, it seems that most standards employ codes that have been shown to be useful in general-purpose communication settings. Although this is justifiable from a pragmatic point of view, we observe that a thorough \emph{information-theoretic approach} may reveal alternate coding schemes that, in general, can provide benefits in terms of reliability, efficiency, synchronization, simplicity, or security.

Operating range of a reader-tag pair is determined by communications requirements as well as power transfer requirements.  To meet the communications requirements, the reader-to-tag and the tag-to-reader communication channels satisfy specified demands on communication transfer rate and reliability. To meet the power transfer requirements, the received power at the tag must be sufficiently large as to provide operating power at the tag.

According to \cite{nik06b,RFID2010}, with current technology it is the power transfer requirements that present the bottleneck with respect to operating range for a two-way reader-tag communication session. Nevertheless, there is a value in determining the information-theoretic aspects, such as tradeoffs between reliability and transmission rate, of this communication: First, because future technologies may shift the relation between communication and power transfer requirements, and second, because present cheap tag technologies impose challenges on communication which are not directly related merely to received signal power.

Wireless information and power transfer has been considered in different contexts previously, for instance, for multiuser orthogonal frequency division multiplexing  systems \cite{zhou13} and cellular networks \cite{hua12}. See also \cite{liu13} and references therein.
In \cite{gro10}, wireless information and power transfer across a noisy inductively coupled channel  was considered from a different perspective than we do in this paper, i.e., it was not considered from the perspective of  code design, but from a \emph{circuit} perspective. For details, we refer the interested reader to \cite{gro10}.  In \cite{var12}, a coding-based secure communication protocol for inductively coupled communication, inspired by quantum key distribution, was recently proposed.



In this paper, however, we address issues related to lower layer coding of information on inductively coupled channels, with emphasis on \emph{coding for error control}  for the \emph{reader-to-tag} channel. The remainder of the paper is organized as follows. In Section~\ref{sec:model}, we describe the characteristics  of the reader-to-tag channel and discuss power issues and processing capabilities.   
A  discretized Gaussian shift channel as a model for the reader-to-tag channel for passive near-field RFID is proposed in Section~\ref{sec:channelmodels}. This model is relevant if the receiver resynchronizes its internal clock each time a bit is detected, and is different from the recently proposed bit-shift channel model in \cite{ros11,ros09glo}.
In Section~\ref{sec:capacity}, we numerically consider its capacity, and, in Section~\ref{sec:coding}, we present several new and very simple codes for this channel model, as well as their encoding/decoding techniques. Simulation results are presented in Section~\ref{sec:sim}, and we draw some conclusions 
 in Section~\ref{sec:conclusion}.


\section{Characteristics of the Reader-to-Tag Channel}
\label{sec:model}

In this paper, we will be concerned with data transfer from a reader to a tag. An \emph{information source} generates an information \emph{frame} of $k$ bits ${\mathbf u} = (u_1,\ldots,u_k)$. The information frame is passed through an encoder to produce an encoded frame ${\mathbf c} = (c_1,\ldots,c_n)$. The encoded frame is interpreted as a waveform that modulates a carrier wave, as shown in  Fig.~\ref{fig:circuit}, \cite{gregory,nik07}.

Please observe that the concept of a frame in this context refers to a collection of bits that belong together, for some semantic reason related to the  application layer. The actual encoder may work at a different length. Due to the strictly limited computing power of the tag, the actual encoder may work on a bit-by-bit basis, as in most of the examples later in this paper. The encoded frame length $n$ may be fixed, depending only on $k$, or variable, depending on $k$ and also on the information frame, but in general $ n \geq k$.

\begin{figure}[htbp]
\centerline{\includegraphics[width=9cm,keepaspectratio=true,angle=0]{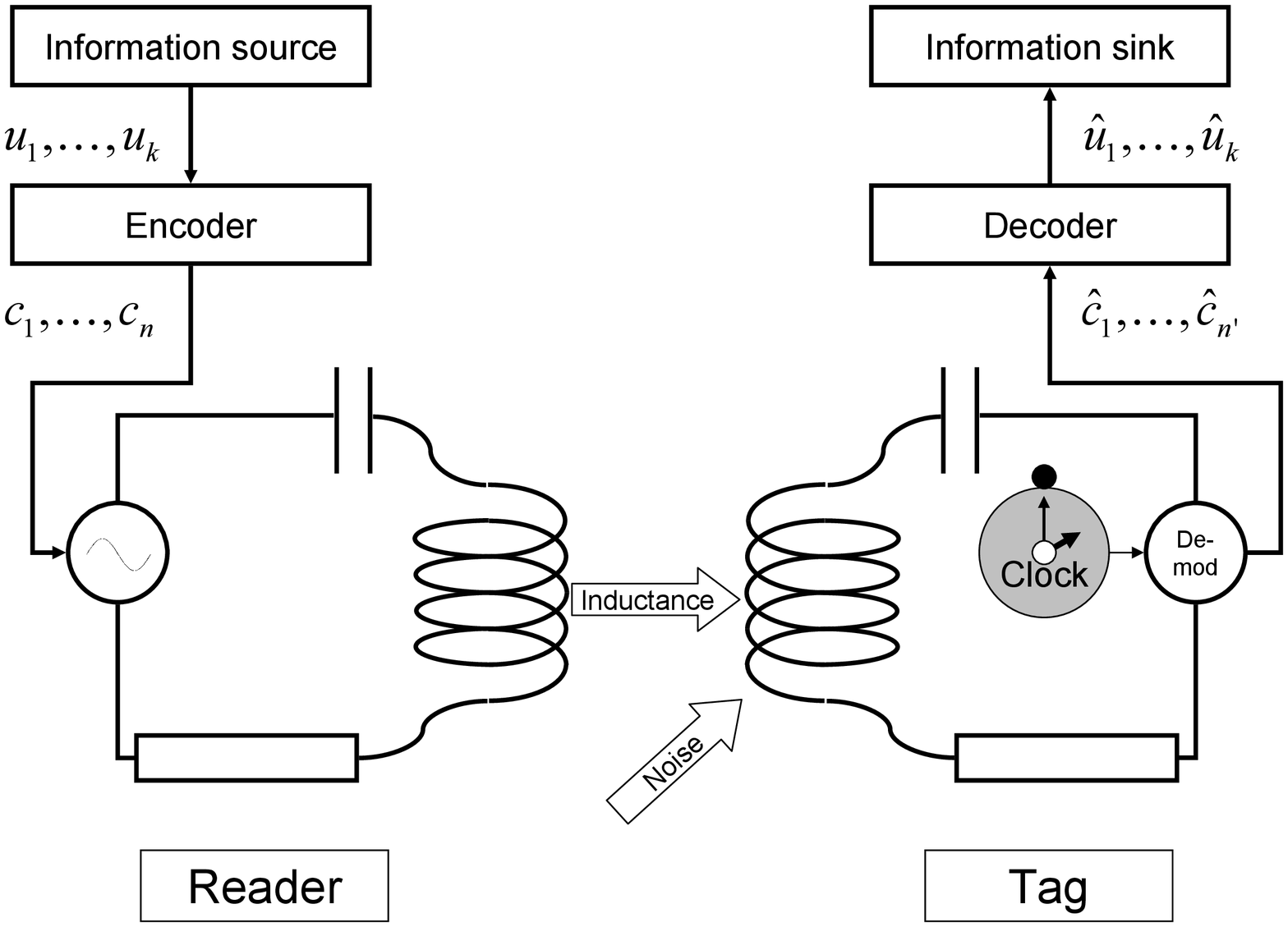}}
\caption{A simplified view of the reader-to-tag channel.}\label{fig:circuit}
\end{figure}

Meanwhile, back at  Fig.~\ref{fig:circuit}, the demodulator in the tag samples  the physical waveform at time intervals determined by the tag's timing device, and converts it into an estimate ${\hat{\mathbf{c}}} = (\hat{c}_1,\ldots,\hat{c}_{n'})$ of the transmitted frame, where in general $n' \neq n$. Ideally, ${\hat{\mathbf{c}}}$ should be identical to ${\mathbf c}$, but additive noise, interference, timing inaccuracies, and waveform degradation due to limited bandwidth may contribute to corrupt the received frame ${\hat{\mathbf{c}}}$. We will discuss some of these signal corruptions later in this paper. A decoder at the tag  subsequently attempts to recover an information frame ${\hat{\mathbf{u}}} = (\hat{u}_1,\ldots,\hat{u}_k)$ from ${\hat{\mathbf{c}}}$. Correct decoding is achieved if  ${\hat{\mathbf{u}} = {\mathbf{u}}}$.

\subsection{Power Issues}
\label{sec:power}
The tag in Fig.~\ref{fig:circuit} has no internal power source. Rather, it collects the power derived from the carrier. After some initial transient delay, the tag's power circuitry will be charged sufficiently to provide operating power for the tag. Commonly, amplitude modulation, or more precisely \emph{on-off keying (OOK)} is employed. In OOK, a "$1$" (resp.\ "$0$") is transmitted by the presence (resp.\ absence, or alternatively a low amplitude) of the carrier for the duration specified for transmitting that particular bit.

The transmitted power is limited by regulation \cite{RFID2010}. However, the amount of transferred power can still be influenced by the encoding scheme used. Although the tag has no traditional battery or other means of accumulating energy over an extended period, it is possible to "ephemerally" store energy over a short time (say, a few bit periods) in the power circuitry.
Thus, it makes sense to impose constraints on power content in the transmitted signal  \cite{ros11,ros09glo,bar08}, for example, by demanding that $m_P$ out of every $n_P$ consecutive transmitted bits are $1$'s. Thus, a high power content (i.e., the ratio $m_P/n_P$ is large) is an advantage. The precise manifestation of this advantage depends on technology and is difficult to measure. Therefore, we will consider different measures of power (to be defined below) as a figure of merit for a given coding scheme.

Formally, we will define the  \emph{power content} of a binary vector $\mathbf{a} \in {\rm GF}(2)^n$, denoted by $P(\mathbf{a})$, as the rational number $w(\mathbf{a})/n$, where $w(\cdot)$ denotes the Hamming weight of its binary argument.

Let $\mathcal{C}$ denote a block code or a variable-length code, i.e., a collection or set of codewords. Furthermore, let $\mathcal{C}^{[N]}$ be the set of sequences of length $N \geq 1$ over $\mathcal{C}$, i.e., the set of $N$ consecutive codewords. The \emph{average power} of $\mathcal{C}$ is defined as the average power content of the sequences in $\mathcal{C}^{[N]}$ as $N \to \infty$. For block codes, this average does not depend on $N$, and the average power of a block code $\mathcal{C}$ is $P_{\rm avg}(\mathcal{C})= \frac{1}{|\mathcal{C}|} \sum_{\mathbf{a} \in \mathcal{C}} P(\mathbf{a})$. However, for variable-length codes, the average depends on $N$, and we need to consider the limit as $N \to \infty$. In general, the average power of a code $\mathcal{C}$ can be computed from \cite{ros11}
\begin{displaymath}
P_{\rm avg}(\mathcal{C}) = \frac{\sum_{j=1}^{|\mathcal{C}|} w_j}{\sum_{j=1}^{|\mathcal{C}|} n_j}
\end{displaymath}
where $w_j$ and $n_j$ denote the Hamming weight and length of the $j$th codeword in $\mathcal{C}$, respectively.

The \emph{minimum sustainable power} of a block or variable-length code $\mathcal{C}$ is defined as $P_{\rm min}(\mathcal{C})=\min_{\mathbf{a} \in \mathcal{C}} P(\mathbf{a})$. 
We remark that for codes defined by a state diagram, the various notions of power can refer to any cycle in the state diagram. Thus, $P_{\rm min}$  refers to the minimum  average cycle weight of a cycle in the state diagram. 

As a final figure of merit, we will consider the \emph{local minimum power} of a code $\mathcal{C}$ as the minimum positive value of the ratio $m_P/n_P$ over all possible sequences in $\mathcal{C}^{[N]}$, for any finite value of $N$, where $n_P \geq m_P$ are arbitrary positive integers.


\subsection{Processing Capability}

Due to the limited tag power, processing capability is severely limited in a tag. This applies to any processing involved in whatever service the tag is supposed to provide, but also signal processing involved in receiving information.

\subsubsection{Error Avoidance Versus Error Correction}

For many communication channels studied in the literature, approaching channel capacity (or even achieving a significant coding gain over naive implementations) relies on error correction codes. However, although classes of codes are known for which the decoder can be efficiently implemented, the decoding process may still require a significant amount of processing. We will argue below that for channels for which the error probabilities depend on the transmitted data, reliability can be increased by using a code enforcing an appropriate set of modulation constraints. Such \emph{error avoiding} codes can typically be decoded by a simple table, mapping received sequences into information estimates.

\section{The Discretized Gaussian Shift Channel}
\label{sec:channelmodels}


In this section, we will discuss a new channel model for the reader-to-tag channel, which is slightly different from the bit-shift model (for inductive coupling) recently introduced in \cite{ros11,ros09glo}. 

If the receiver resynchronizes its internal clock each time a bit is detected, the bit-shift model from \cite{ros11,ros09glo} needs to be modified. We will first introduce the \emph{Gaussian shift channel}.

Suppose the reader transmits a run of $\tilde{x}$ consecutive equal symbols. This corresponds to an amplitude modulated signal of duration $\tilde{x}$. At the tag, we will assume that this is detected (according to the tag's internal clock) as having duration $\tilde{y}$, where
\begin{equation}
\label{eq:runlFormat1}
\tilde{y} = \tilde{x} \cdot K
\end{equation}
and $K$ is a random variable. In this paper,  $K$  follows a Gaussian distribution  $\mathbf{N}(\nu, \varepsilon^2)$ with mean $\nu$ and variance $\varepsilon^2$. Consecutive samplings of $K$ are assumed to be independent. If $\nu \neq 1$, it means that the tag has a systematic drift, which may affect the tag's ability to function at all. Thus, we will focus on the case $\nu = 1$. With this assumption, the input to the demodulator will be a sequence of alternating runs of high and low amplitude values; the detected duration ${\tilde{y}}$ of each run being a \emph{real-valued} number.


We might attempt decoding directly at the Gaussian shift channel, but the computational complexity may be high for the tag receiver. As a simplification, and to deal with the fact that $\tilde{y}$ may become negative ($K$ has a normal distribution), which of course does not have any physical interpretation,  we propose to discretize the timing and truncate $K$. The optimal choice for the quantization thresholds, i.e., the thresholds when mapping the real-valued numbers ${\tilde{y}}$ to \emph{positive} integers, will depend on the code under consideration.  However, an optimal \emph{local} threshold can be computed as shown in the following proposition.
\begin{proposition} \label{prop:1}
Let $a$ and $b$ be  positive integers with $b > a$, representing the only two legal runlengths in a given constrained code. Then, there is a single threshold $t = t(a,b)$, and its optimum value from a \emph{local} perspective\footnote{We can do better with a maximum-likelihood (ML) detector which considers the whole transmitted sequence.}  to determine if runlength $a$ or runlength $b$ was transmitted is
\begin{displaymath}
t = t(a,b) = \frac{2ab}{a+b}.
\end{displaymath}
The corresponding decision error with one such decision is
\begin{equation} \label{eq:decision_error}
Q\left(\frac{t-a}{a \varepsilon}\right) = Q \left( \frac{b-a}{(a+b) \varepsilon} \right) > Q \left(\frac{1}{\varepsilon} \right)
\end{equation}
where $Q(x)$ is the probability that a sample of the standard normal distribution has a value larger than
$x$ standard deviations above the mean, i.e.,
\begin{equation} \notag 
    Q(x) = \int_x^{\infty}\frac{1}{\sqrt{2\pi}}e^{-y^2/2}dy = \frac{1}{2}\erfc\left(\frac{x}{\sqrt{2}}\right)
\end{equation}
where $\textrm{erfc}(\cdot)$ denotes the complementary error function.
\end{proposition}

\begin{IEEEproof}
Assuming $a$ is transmitted, then the probability that $b$ is received (with $t$ as the quantization threshold) is $Q((t-a)/ a \varepsilon)$. Likewise, if $b$ is transmitted, then the probability that $a$ is received is $Q((b-t)/ b \varepsilon)$. This follows directly from the fact that $K$ has a Gaussian distribution with mean $1$ and variance $\varepsilon^2$. Now, since $Q(\cdot)$ is a monotonically decreasing function, $Q((t-a)/ a \varepsilon)$ is monotonically decreasing and  $Q((b-t)/ b \varepsilon)$ is  monotonically increasing in $t$ (within the range $[a,b]$). Thus, the optimal threshold $t$ corresponds to the intersection of $Q((t-a)/ a \varepsilon)$ and $Q((b-t)/ b \varepsilon)$. Thus, 
$(t-a)/a \varepsilon = (b-t)/b \varepsilon$. Solving this equation, we get $t = 2ab/(a+b)$.
The expression for the decision error in (\ref{eq:decision_error}) follows by substituting the expression for the optimal threshold $t$ into either $Q((t-a)/ a \varepsilon)$ or $Q((b-t)/ b \varepsilon)$, and the final inequality (in (\ref{eq:decision_error})) follows from the fact that $(b-a)/(a+b)$ is smaller than  $1$.
\end{IEEEproof}
%
%
%

Note that when $a=b-1$, $t=2ab/(a+b)=2b(b-1)/(2b-1)$ will approach $(a+b)/2=b-1/2$ as $b$ goes to infinity.

 We remark that we do not allow the mapping of a real-valued number (from the output of the Gaussian shift channel) to zero (or a negative integer), which means that the channel can not make a runlength disappear. This appears to be consistent with properties of practical inductively coupled channels.

In general, let $\mathcal{Q}(\mathcal{A},\mathcal{T})$ denote a quantization scheme with quantization values  $\mathcal{A}=\{a_1,\ldots,a_{|\mathcal{A}|}\}$, where $1 \leq a_1 < \cdots < a_{|\mathcal{A}|} \leq L$, and $L$ is some positive integer (that later will be used as a \emph{runlength}), and quantization thresholds $\mathcal{T}=\{t_2,\ldots,t_{|\mathcal{A}|}\}$, where $a_{l} < t_{l+1} < a_{l+1}$, $l=1,\ldots,|\mathcal{A}|-1$. The quantization scheme works in the following way. Map a received real-valued number to an integer in $\mathcal{A}$ using quantization thresholds in $\mathcal{T}$, i.e., if the received real-valued number is in the range $[ t_l,t_{l+1})$, $l=2,\ldots,|\mathcal{A}|-1$, map it to $a_{l}$, if it is in the range $[t_{|\mathcal{A}|},\infty)$, map it to $a_{|\mathcal{A}|}$, and, otherwise, map it to $a_1$.

Now, we define the discretized Gaussian shift channel with quantization scheme $\mathcal{Q}(\mathcal{A},\mathcal{T})$ as the cascade of the Gaussian shift channel and the quantization scheme $\mathcal{Q}(\mathcal{A},\mathcal{T})$, where the quantization scheme $\mathcal{Q}(\mathcal{A},\mathcal{T})$ is applied to the real-valued sequence at the output of the Gaussian shift channel.

As an example, we can define a discretized Gaussian shift channel, where the quantization thresholds are chosen such that the integer sequence is obtained from the real-valued sequence by rounding its values to the nearest positive integer value. This particular quantization scheme will be denoted by $\mathcal{Q}_{\rm rounding}$. As a further modification, we may introduce a parameter $\Gamma$, to truncate the maximum observed length, into the quantization scheme  $\mathcal{Q}_{\rm rounding}$, and in this way get a family of discretized Gaussian shift channels. The resulting quantization scheme works in the following way.
%
%
If the reader has transmitted a run of $L$ symbols, the tag will detect it as having length
\begin{displaymath}
\begin{cases}
L-l, & \text{if $K-1 \in \left[-\frac{2l+1}{2L},-\frac{2l-1}{2L}\right)$ and} \\
& l=1,\ldots,\Gamma'-1 \\
L-\Gamma', & \text{if $K-1 \in \left(-\infty,-\frac{2\Gamma'-1}{2L} \right)$} \\
L, & \text{if $K-1 \in \left[-\frac{1}{2L},\frac{1}{2L} \right)$}\\
L+l, & \text{if $K-1 \in \left[\frac{2l-1}{2L},\frac{2l+1}{2L} \right)$ and} \\
& l=1,\ldots,\Gamma-1 \\
L+\Gamma, & \text{if $K-1 \in \left[\frac{2\Gamma-1}{2L},\infty \right)$}
\end{cases}
\end{displaymath}
where $\Gamma \geq 1$ is a truncation integer parameter and $\Gamma'=\min(\Gamma,L-1)$.
%
With $\Gamma=1$, we denote the channel as the discretized Gaussian single-shift channel. 
With $\Gamma=2$, the channel is called the discretized Gaussian double-shift channel, and so on.  
%
%
%
%
%
%
Now, if we want to express the  discretized Gaussian single-shift channel 
in terms of runlengths with additive error terms $\omega_i$ (as in \cite[Eq.\ (4)]{ros11}),  \cite[Eq.\ (4)]{ros11} is modified by (\ref{eq:runlFormat1})  and discretization to
\begin{equation}  \label{eq:model}
\tilde{y}_i = \tilde{x}_i + \omega_i 
\end{equation}
where
\begin{equation} \label{eq:tran_prob_rounding}
\begin{split}
{P}(\omega_i=\omega | \tilde{x}_i=\tilde{x})
= \begin{cases}
p(\tilde{x}), & \text{if $\omega=-1$ and $\tilde{x}>1$}\\
0, & \text{if $\omega=-1$ and $\tilde{x}=1$}\\
1-2p(\tilde{x}), & \text{if $\omega=0$ and $\tilde{x}>1$}\\
1-p(\tilde{x}), & \text{if $\omega=0$ and $\tilde{x}=1$}\\
p(\tilde{x}), & \text{if $\omega=1$ and $\tilde{x} \geq 1$}\\
0, & \text{otherwise} \end{cases}
\end{split}
\end{equation}
and $p(L) = Q \left(\frac{1}{2L\varepsilon} \right)$.


As another example, we can define a quantization scheme $\mathcal{Q}(\mathcal{A}) = \mathcal{Q}(\mathcal{A},\mathcal{T})$, where  the quantization threshold $t_l=2a_{l-1}a_{l}/(a_{l-1}+a_{l})$, $l=2,\ldots,|\mathcal{A}|$. In a similar manner, as for $\mathcal{Q}_{\rm rounding}$, we can express the  discretized Gaussian single-shift channel (now with quantization scheme $\mathcal{Q}(\mathcal{A})$) 
in terms of runlengths with additive error terms $\omega_i$ as in (\ref{eq:model}), but with transition probabilities
\ifonecolumn
\begin{equation} \label{eq:tran_prob_q}
{P}(\omega_i=\omega | \tilde{x}_i=\tilde{x})
= \begin{cases}
p(\alpha(\tilde{x})), & \text{if $\omega=-1$ and $\tilde{x}>1$}\\
0, & \text{if $\omega=-1$ and $\tilde{x}=1$}\\
1-p(\alpha(\tilde{x}))-p(\beta(\tilde{x})), & \text{if $\omega=0$ and $\tilde{x}>1$}\\
1-p(\beta(\tilde{x})), & \text{if $\omega=0$ and $\tilde{x}=1$}\\
p(\beta(\tilde{x})), & \text{if $\omega=1$ and $\tilde{x} \geq 1$}\\
0, & \text{otherwise} \end{cases}
\else
\begin{equation} \label{eq:tran_prob_q}
\begin{split}
&{P}(\omega_i=\omega | \tilde{x}_i=\tilde{x})\\
&= \begin{cases}
p(\alpha(\tilde{x})), & \text{if $\omega=-1$ and $\tilde{x}>1$}\\
0, & \text{if $\omega=-1$ and $\tilde{x}=1$}\\
1-p(\alpha(\tilde{x}))-p(\beta(\tilde{x})), & \text{if $\omega=0$ and $\tilde{x}>1$}\\
1-p(\beta(\tilde{x})), & \text{if $\omega=0$ and $\tilde{x}=1$}\\
p(\beta(\tilde{x})), & \text{if $\omega=1$ and $\tilde{x} \geq 1$}\\
0, & \text{otherwise} \end{cases}
\end{split}
\fi
\end{equation}
where
\begin{equation} \label{eq:alphabeta}
\alpha(\tilde{x}) = \frac{\tilde{x}_{\rm previous} + \tilde{x}}{2(\tilde{x} - \tilde{x}_{\rm previous})} \text{ and } \beta(\tilde{x}) = \frac{\tilde{x}_{\rm next} + \tilde{x}}{2(\tilde{x}_{\rm next} - \tilde{x})}
\end{equation}
and where $\tilde{x}_{\rm previous}$ (resp.\ $\tilde{x}_{\rm next}$)  is the closest value to $\tilde{x}$ allowed by the quantization scheme that is also strictly smaller (resp.\ larger) than $\tilde{x}$.

As will become clear later, this quantization scheme outperforms the general rounding scheme defined above.  However, note that when $a_{|\mathcal{A}|-1} = a_{|\mathcal{A}|}-1$ and $a_{|\mathcal{A}|}$ is large, the performance approaches the performance of the discretized Gaussian shift channel with quantization scheme $\mathcal{Q}_{\rm rounding}$ for low values of $\varepsilon$.


We can make the following remarks in connection with the Gaussian shift channel.
\begin{enumerate}[(i)]
  \item As can be seen from Fig.~\ref{fig:Qx}, when considering the ``likely error patterns'', we need to be concerned mainly about the longest runs of equal symbols. The exception to this pragmatic rule occurs when, for some codes, it is possible to correct all shifts (up to some order, where a single shift is a shift of order one, a double shift is a shift of order two, and so on) corresponding to  maximum-length runs.

\begin{figure}[!hbp]
\centerline{\includegraphics[width=8cm,keepaspectratio=true,angle=0]{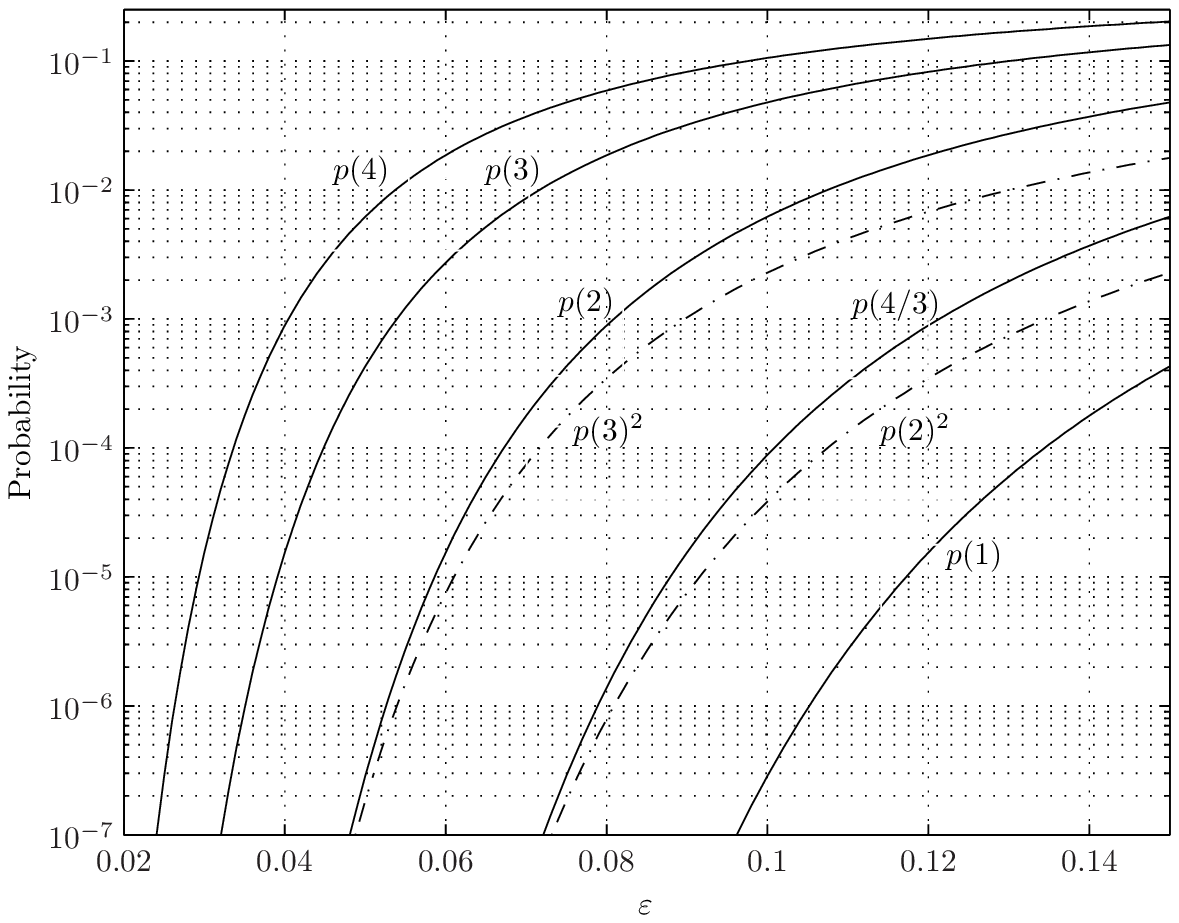}}
\caption{Comparison of shift probabilities (with quantization scheme $\mathcal{Q}_{\rm rounding}$) versus $\varepsilon$ for runlengths $1$, $2$, $3$, and $4$.}
\label{fig:Qx}
\end{figure}

\item For \emph{many simple codes} used on the discretized Gaussian shift channel the frame error rate (FER), denoted by $P(\mbox{FE})$,  can be simplified to, respectively,

\begin{equation} \notag
  P(\mbox{FE}) \approx S_L \cdot p(L)
\end{equation}
and
\begin{equation} \notag
  P(\mbox{FE}) \approx S_L \cdot p(L-1/2)
\end{equation}
with quantization schemes $\mathcal{Q}_{\rm rounding}$ and $\mathcal{Q}(\mathcal{A})$, where $\mathcal{A}=\{a_1,\ldots,a_{|\mathcal{A}|-2},L-1,L\}$, and
$S_L$ is some constant (representing a count of different error events) depending on the code and on the specific decoder, assuming that the most likely error event when using the code is connected with the confusion of runlengths of length $L$ with some other run of length $L-1$. We omit the details, but will show examples later (see Theorems~\ref{thm:RLL} and \ref{thm:RLLnew}).
 \item Error avoidance versus error control: Suppose we can design an error correction encoder that admits runlengths of length at most 2; that has a decoder that can correct all error events involving a single shift of a \emph{single} run of length 2, but that will make a mistake if two or more such  event occurs. Such a decoder should have a FER on the order of $p(2)^2$  (with quantization scheme $\mathcal{Q}_{\rm rounding}$) for small $\varepsilon$. Observe from Fig.~\ref{fig:Qx} that $p(2)^2 > p(1)$. Can we design a  code with a  simple decoder that behaves as $p(1)$? Yes, we can; see Sections~\ref{sec:13:13}, \ref{sec:1:13}, and \ref{sec:relatedconstraints}.
  \item Observe that the discretized Gaussian single-shift channel is a \emph{special form of an insertion-deletion channel}, which randomly may extend or shorten the runs of transmitted identical symbols, but where the statistics of this random process depend on the length of the runs.
Codes for insertion-deletion channels have been studied, but to a moderate extent, and some of the best known codes, such as the Varshamov-Tenengolz codes \cite{klo95} and the codes in \cite{LiuMitz2010}, are apparently too complex for the application in question and also do not possess the appropriate modulation constraints, to be discussed below.
  \item An intelligent receiver tag should realize that any received run longer than the maximum run must be the result of an insertion. Thus, such insertions can trivially be corrected. In consequence, \emph{for some codes}, the discretized Gaussian shift channel is approximately simply a special deletion channel that applies only to runs of maximum length.
  \item In general, for any code and channel, a receiver may use a forward error correction scheme (FEC), or an automatic-repeat-request (ARQ) scheme asking for retransmissions if an error is detected. Obviously, error detection is computationally simpler than error correction. Indeed, ARQ is typically used in standard protocols for passive RFID, utilizing a standard embedded cyclic redundancy check code.

For the binary symmetric channel it is further well-known that the FER associated with FEC is typically much higher than the probability of undetected error corresponding to ARQ. Counter-intuitively, this property does not necessarily apply with the discretized Gaussian shift channel.
%
\end{enumerate}

\section{Channel Capacity} \label{sec:capacity}

In this section, we will consider the channel capacity of the discretized Gaussian shift channel.

Since a sequence of transmitted consecutive bits can not disappear (the quantization schemes quantize each real-valued received number to a positive integer) and consecutive samplings of $K$ are independent, the discretized Gaussian shift channel (with any quantization scheme) is really a discrete memoryless channel operating on runlengths with the positive integer values as input and output alphabet, and with  channel transition probabilities that depend on $\varepsilon$ and the quantization scheme. Now, we define a truncated version of the channel, denoted by $\mathcal{H}_{L,T}$, with input alphabet $\mathcal{X}=\{1,\dots,L\}$, output alphabet $\mathcal{Y}=\{1,\dots,L'\}$, where $L$ and $L'$ are integers greater than one, and channel transition probabilities $f_{Y|X}(y|x)$. The parameter $L'$ is the smallest integer output of the discretized Gaussian shift channel (with a given quantization scheme) such that the probability of observing $L'$ for any given input $x \in \mathcal{X}$ is smaller than some small threshold probability $T$. The normalized mutual information between the channel input $X$ and channel output $Y$, denoted by $\tilde{I}(X;Y)$ and measured in bits per input symbol, can be expressed by \cite[Eq.\ (3)]{sha91}


\ifonecolumn
\begin{equation} \label{34:eq:I}
\tilde{I}(X;Y) = \frac{I(X;Y)}{\mathbb{E}[X]} = \frac{\sum_{y \in \mathcal{Y}} \sum_{x \in \mathcal{X}} f_X(x) f_{Y|X}(y|x) \log_2 \left( \frac{f_{Y|X}(y|x)}{\sum_{j \in \mathcal{X}} f_X(j)f_{Y|X}(y|j)} \right)}{\sum_{j \in \mathcal{X}} j \cdot f_X(j)}
\end{equation}
\else
\begin{equation} \label{34:eq:I}
\begin{split}
&\tilde{I}(X;Y) = \frac{I(X;Y)}{\mathbb{E}[X]} \\
&= \frac{\sum_{y \in \mathcal{Y}} \sum_{x \in \mathcal{X}} f_X(x) f_{Y|X}(y|x) \log_2 \left( \frac{f_{Y|X}(y|x)}{\sum_{j \in \mathcal{X}} f_X(j)f_{Y|X}(y|j)} \right)}{\sum_{j \in \mathcal{X}} j \cdot f_X(j)} 
\end{split}
\end{equation}
\fi
where $I(X;Y)$ denotes the mutual information between $X$ and $Y$ and $\mathbb{E}[X]$ the expectation of $X$ with respect to the input probability distribution $f_X(x)$. Now,
the capacity of $\mathcal{H}_{L,T}$ (in bits per symbol) can be obtained by maximizing the fraction in (\ref{34:eq:I}) over all input probability distributions $f_X(x)$. Note that since the channel is memoryless, it is sufficient to consider only a single use of the channel, i.e., not sequences of length $N$ as in  \cite[Eq.\ (3)]{sha91}.

We remark that if the channel could in fact remove runlengths, then the channel would resemble a deletion channel with substitution errors operating on runlengths. From  an information-theoretic perspective, such a channel is much harder to analyze, since there is no finite-letter expression for the channel capacity \cite{fer11}. 


From \cite[p.\ 191]{cov06}, we know that the numerator of (\ref{34:eq:I}), i.e., the mutual information between $X$ and $Y$,  is a continuous and concave function of $f_X(x)$. Thus, the maximization of the \emph{un-normalized} mutual information (i.e., the maximization of the numerator of (\ref{34:eq:I}) over the set of all input probability distributions $f_X(x)$) can be done using, for instance,  a gradient ascent algorithm, or the iterative Blahut-Arimoto algorithm \cite{ari72,bla72}.

\begin{proposition}
The normalized mutual information $\tilde{I}(X;Y)$ in (\ref{34:eq:I}) is  quasiconcave in $f_X(x)$.
\end{proposition}

\begin{proof}
Since $\mathbb{E}[X]$ is a linear function on $f_X(x)$, it is obviously a convex function on $f_X(x)$. Furthermore,  since $\mathbb{E}[X]$ is strictly positive, then $1/\mathbb{E}[X]$ is also a convex function on $f_X(x)$. Since $I(X;Y)$ is continuous and concave (from  \cite[p.\ 191]{cov06}), it follows  that the fraction in (\ref{34:eq:I}) is a product of a convex and a concave function. Now, $1/\mathbb{E}[X]$ is actually both quasiconvex and quasiconcave because the upper and lower countoursets are always convex sets, since the level sets are linear varieties  (they are linear for $\mathbb{E}[X]$, and hence they are also linear for $1/\mathbb{E}[X]$, since they have the same level sets). Every concave function is quasiconcave, hence $I(X;Y)$ is quasiconcave. Thus, the normalized mutual information in (\ref{34:eq:I}) is the product of two quasiconcave and nonnegative 
functions, which again is quasiconcave.
\end{proof}

The function in (\ref{34:eq:I}) is continuous in $f_X(x)$ (the denominator is strictly positive and continuous, and the numerator is continuous), which is not a general property of being quasiconcave. Furthermore, any \emph{strong} local maximum  is a global maximum for any quasiconcave function \cite{lue68} (the result is formulated for quasiconvex functions in Lemma~1 in \cite{lue68}). Thus, a gradient ascent algorithm can be used to find the global maximum of any  continuous differentiable quasiconcave function by checking for strict maximality.

For illustration purposes, in Fig.~\ref{34:fig:1}, the normalized  mutual information from (\ref{34:eq:I}) is plotted as a function of $f_X(1)$ and $f_X(2)$ when $L=3$, for $\varepsilon=0.15$. The threshold probability is $T=10^{-8}$ and the quantization scheme $\mathcal{Q}_{\rm rounding}$ has been used. 

\ifonecolumn
\begin{figure}[t!]
\par
\begin{center}
\includegraphics[width=0.5\columnwidth]{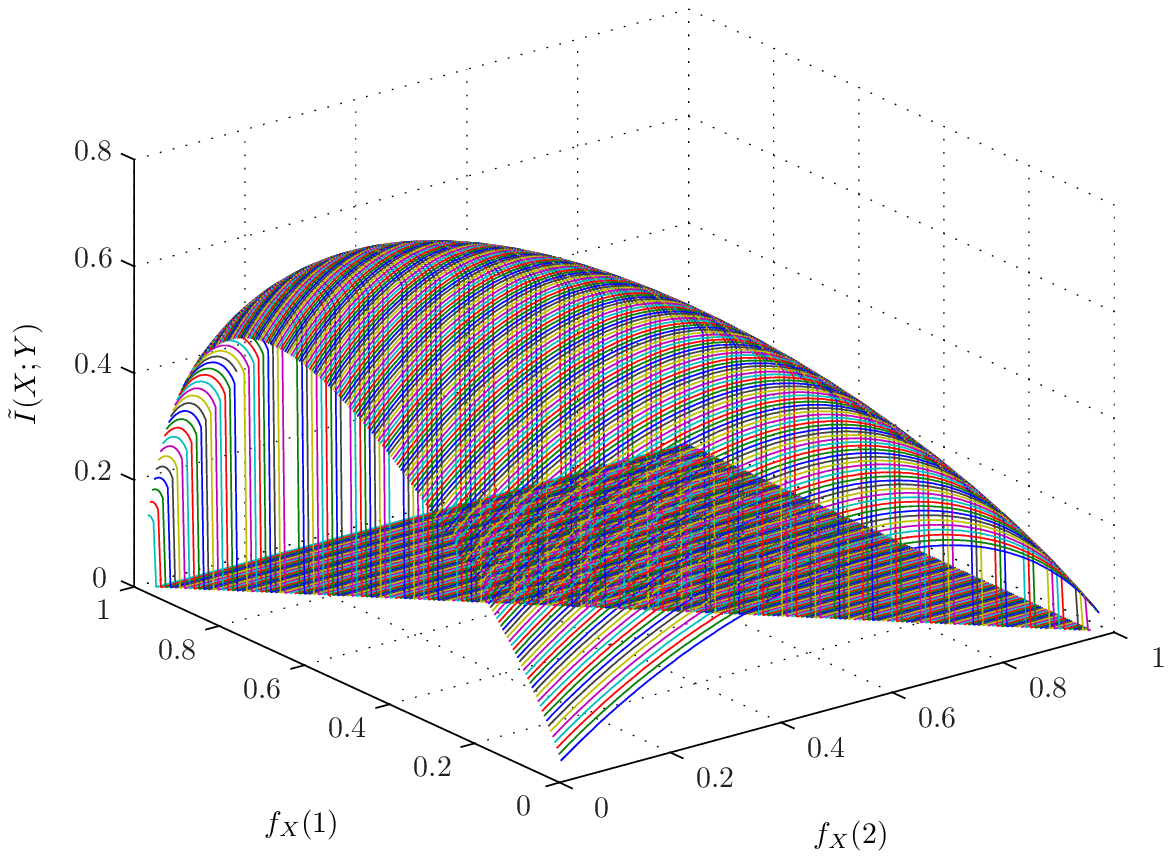}
\end{center}
\caption{\label{34:fig:1} {Normalized mutual information as a function of $f_X(1)$ and $f_X(2)$ when $L=3$ ($f_X(1)+f_X(2)+f_X(3)=1$), for $\varepsilon=0.15$. The threshold probability is $T=10^{-8}$ and the quantization scheme $\mathcal{Q}_{\rm rounding}$ has been used.}}
\end{figure}
\else
\begin{figure}[t!]
\par
\begin{center}
\includegraphics[width=\columnwidth]{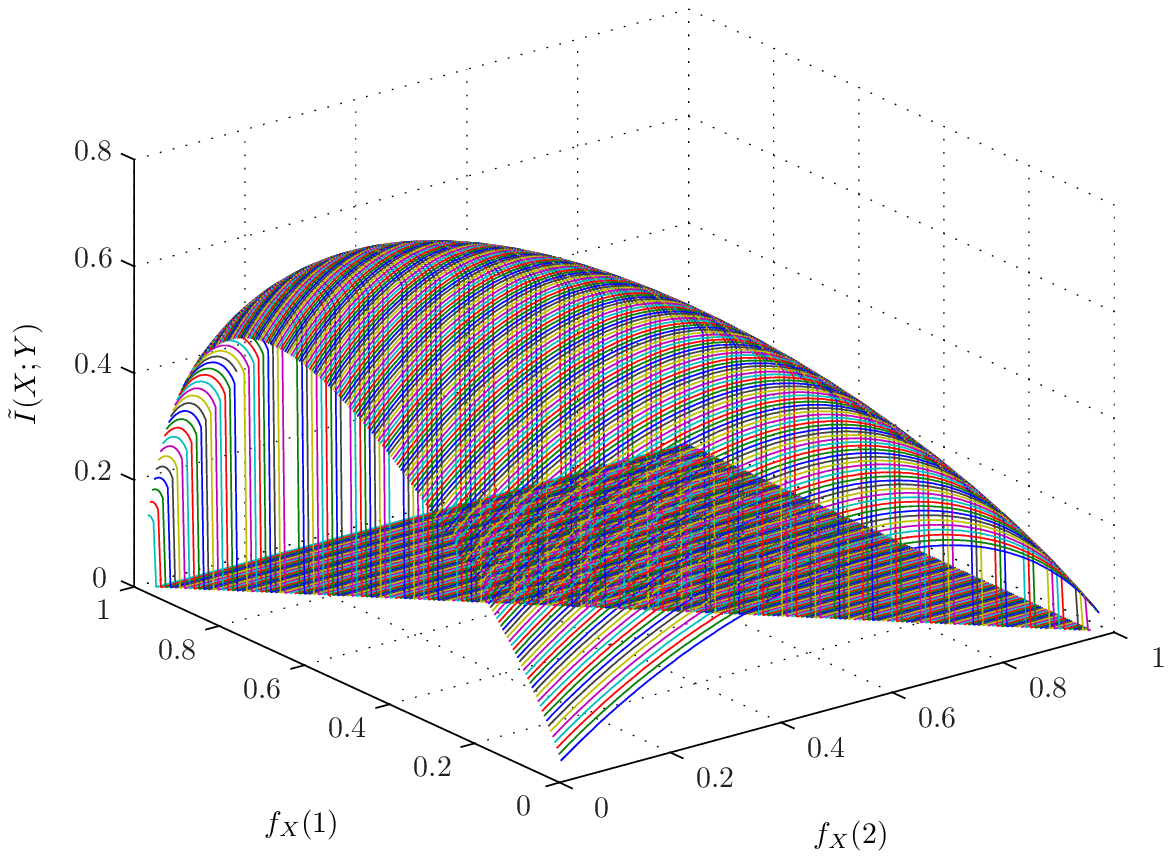}
\end{center}
\caption{\label{34:fig:1} {Normalized mutual information as a function of $f_X(1)$ and $f_X(2)$ when $L=3$ ($f_X(1)+f_X(2)+f_X(3)=1$), for $\varepsilon=0.15$. The threshold probability is $T=10^{-8}$ and the quantization scheme $\mathcal{Q}_{\rm rounding}$ has been used.}}
\end{figure}
\fi

Due to the constraint $\sum_{x \in \mathcal{X}} f_X(x)=1$, the normalized mutual information in (\ref{34:eq:I}) is really a function of $L-1$ variables $f_X(x)$, $x=1,\dots,L-1$. Thus, we may substitute $f_X(L)=1-\sum_{x=1}^{L-1} f_X(x)$ into  (\ref{34:eq:I}) and then compute the partial derivatives with respect to $f_X(x)$, $x=1,\dots,L-1$.

\begin{proposition} \label{34:lem:2}
The partial derivative of  the normalized mutual information $\tilde{I}(X;Y)$ in (\ref{34:eq:I})  with respect to $f_X(x)$, $x=1,\dots,L-1$, is
\begin{equation} \notag
\frac{\partial \tilde{I}(X;Y)}{\partial f_X(x)} = \frac{ \frac{\partial I(X;Y)}{\partial f_X(x)} \sum_{j \in \mathcal{X}} j \cdot f_X(j) - I(X;Y) (x-L)}{ \left(  \sum_{j \in \mathcal{X}} j \cdot f_X(j) \right)^2}
\end{equation}
where
\ifonecolumn
\begin{equation} \notag
\frac{\partial I(X;Y)}{\partial f_X(x)} = \sum_{y \in \mathcal{Y}} f_{Y|X}(y|x) \log_2 \left( \frac{ f_{X|Y}(x|y)}{f_X(x)} \right)
- \sum_{y \in \mathcal{Y}} f_{Y|X}(y|L) \log_2 \left( \frac{f_{X|Y}(L|y)}{f_X(L)} \right).
\end{equation}
\else
\begin{equation} \notag
\begin{split}
\frac{\partial I(X;Y)}{\partial f_X(x)} &= \sum_{y \in \mathcal{Y}} f_{Y|X}(y|x) \log_2 \left( \frac{ f_{X|Y}(x|y)}{f_X(x)} \right) \\
&\;\;\;\;- \sum_{y \in \mathcal{Y}} f_{Y|X}(y|L) \log_2 \left( \frac{f_{X|Y}(L|y)}{f_X(L)} \right).
\end{split}
\end{equation}
\fi
\end{proposition}

\begin{proof}
This follows from straightforward calculus.
\end{proof}

In summary: Recall that the capacity of the channel $\mathcal{H}_{L,T}$ is equal to $\max_{f_X(x)}\tilde{I}(X;Y)$, and that Proposition 2 shows that $\tilde{I}(X;Y)$ is quasiconcave in $f_X(x)$. Thus, in order to numerically determine the capacity of the channel $\mathcal{H}_{L,T}$, we have implemented a steepest ascent method using Proposition~\ref{34:lem:2} for the expression of the gradient. In addition, we need to check if the located stationary point of $\tilde{I}(X;Y)$, i.e.,  a point for which the partial derivatives $\frac{\partial \tilde{I}(X;Y)}{\partial f_X(x)}$ are zero for every $f_X(x)$, indeed corresponds to a strict maximum of $\tilde{I}(X;Y)$. If not, another random starting point for the steepest ascent method is chosen, and the procedure is repeated until a strict maximum of $\tilde{I}(X;Y)$ is located.

\subsection{Numerical Results}

\ifonecolumn
\begin{figure}[tbp]
\par
\begin{center}
\includegraphics[width=0.5\columnwidth]{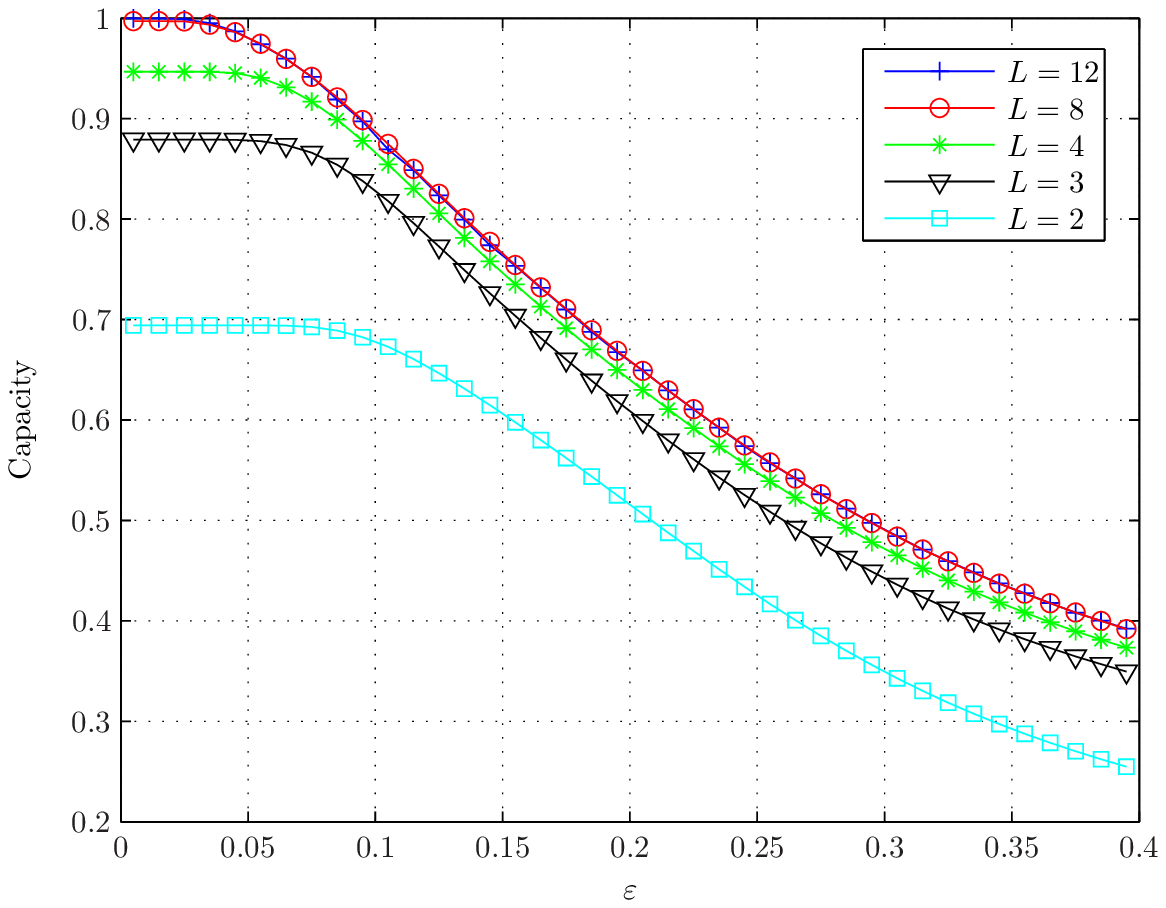}
\end{center}
\caption{\label{34:fig:2} {The capacity of the channel $\mathcal{H}_{L,T}$ for different values of $L$ as a function of $\varepsilon$. The threshold probability is $T=10^{-8}$ and the quantization scheme $\mathcal{Q}_{\rm rounding}$ has been used.}}
\end{figure}
\else
\begin{figure}[tbp]
\par
\begin{center}
\includegraphics[width=\columnwidth]{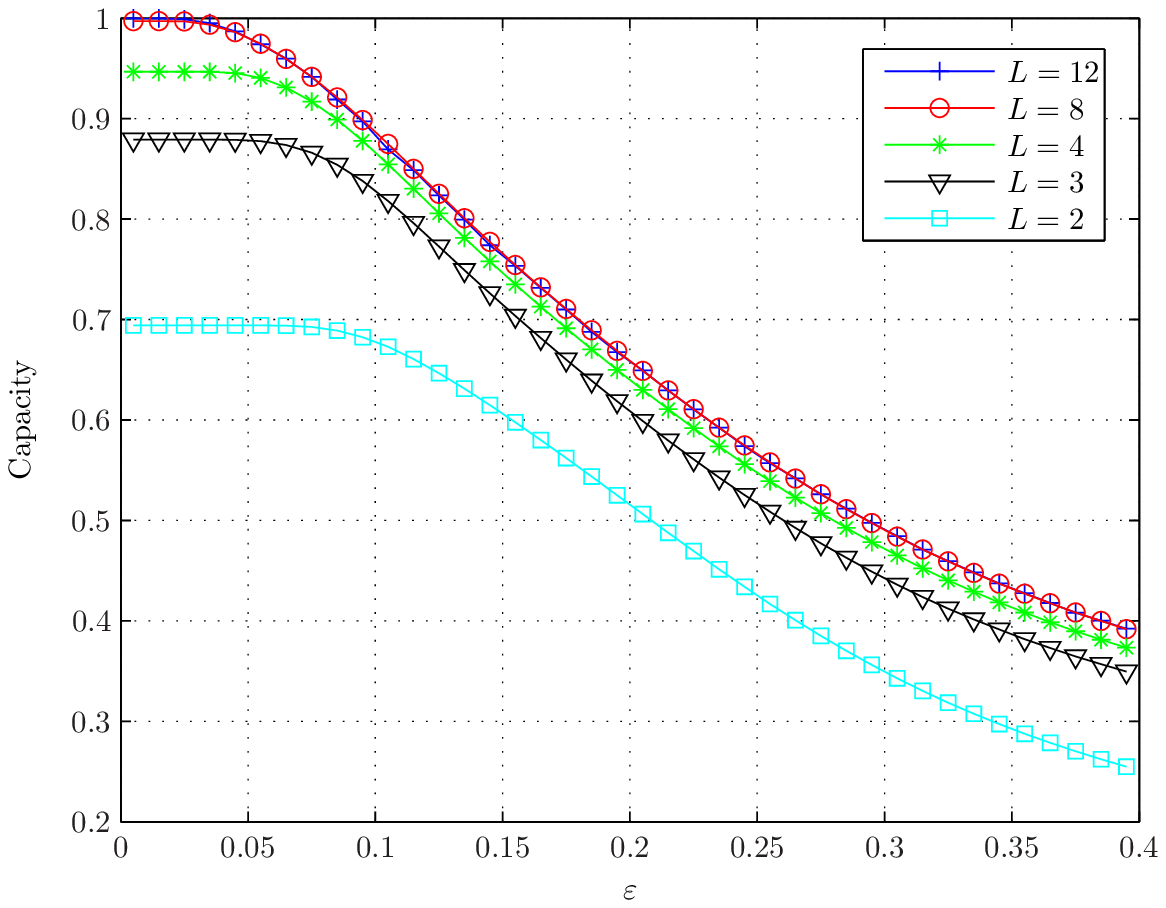}
\end{center}
\caption{\label{34:fig:2} {The capacity of the channel $\mathcal{H}_{L,T}$ for different values of $L$ as a function of $\varepsilon$. The threshold probability is $T=10^{-8}$ and the quantization scheme $\mathcal{Q}_{\rm rounding}$ has been used.}}
\end{figure}
\fi

\ifonecolumn
\begin{figure}[tbp]
\par
\begin{center}
\includegraphics[width=0.5\columnwidth]{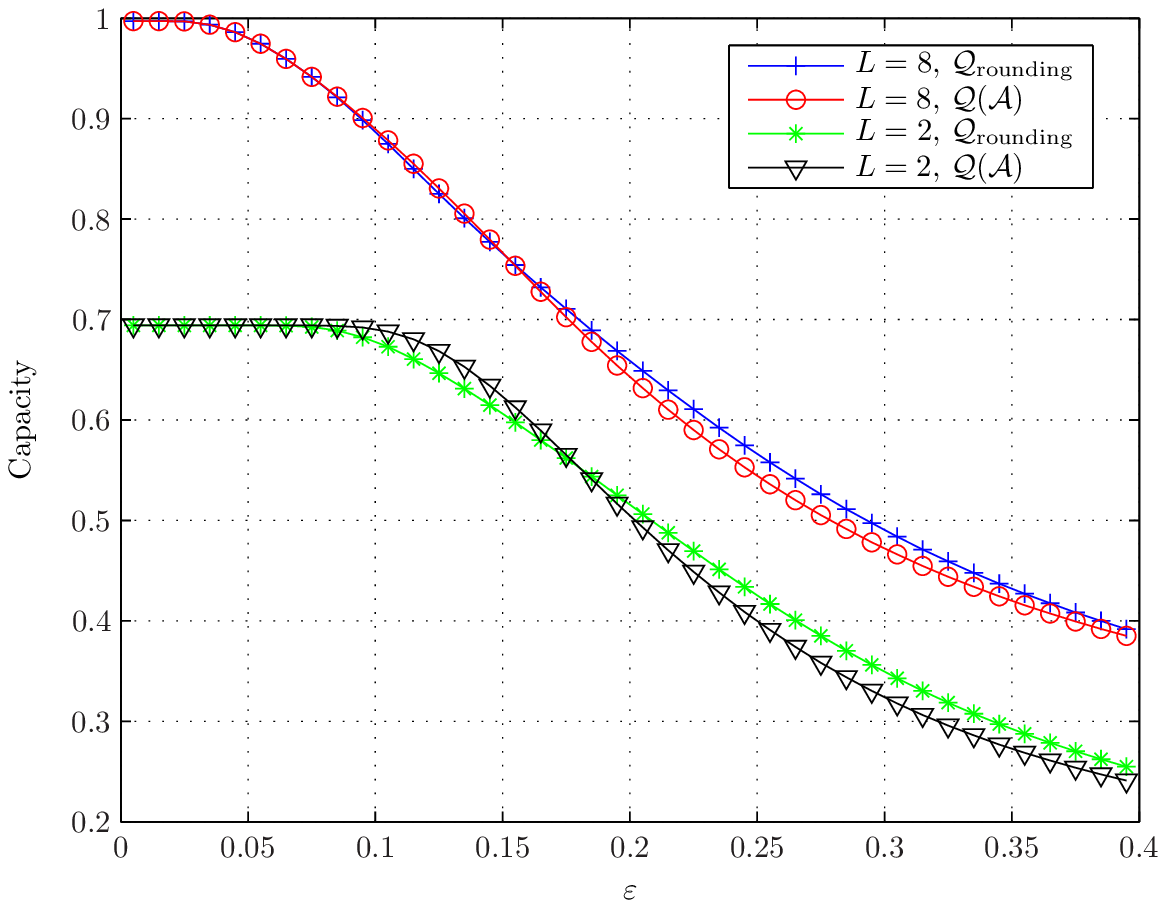}
\end{center}
\caption{\label{34:fig:3} {The capacity of the channel $\mathcal{H}_{L,T}$ for different values of $L$ as a function of $\varepsilon$ with both quantization schemes $\mathcal{Q}_{\rm rounding}$ and $\mathcal{Q}(A)$. The threshold probability is $T=10^{-8}$.}}
\end{figure}
\else
\begin{figure}[tbp]
\par
\begin{center}
\includegraphics[width=\columnwidth]{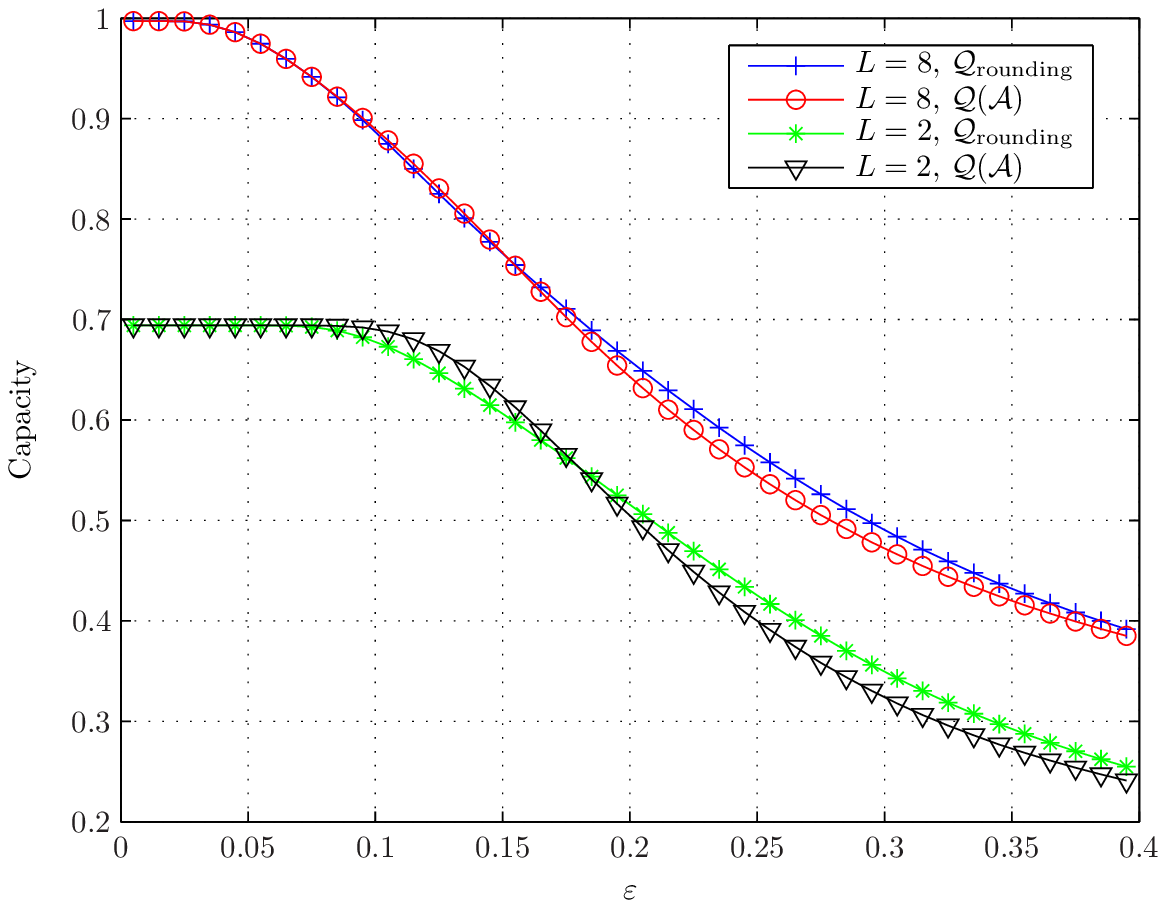}
\end{center}
\caption{\label{34:fig:3} {The capacity of the channel $\mathcal{H}_{L,T}$ for different values of $L$ as a function of $\varepsilon$ with both quantization schemes $\mathcal{Q}_{\rm rounding}$ and $\mathcal{Q}(A)$. The threshold probability is $T=10^{-8}$.}}
\end{figure}
\fi

In Fig.~\ref{34:fig:2}, we have plotted the capacity of the channel $\mathcal{H}_{L,T}$ as a function of $\varepsilon$ for various values of the input alphabet size $L$. The threshold probability is $T=10^{-8}$ and the quantization scheme $\mathcal{Q}_{\rm rounding}$ has been used. The curves in Fig.~\ref{34:fig:2} are computed using a gradient ascent algorithm using the gradient from Proposition~\ref{34:lem:2}. 
We observe that there is almost no difference between the curves for $L=8$ and $L=12$, which indicates convergence. Thus, the curve for $L=12$ should be very close to the \emph{exact} capacity of the discretized Gaussian shift channel with quantization scheme $\mathcal{Q}_{\rm rounding}$.

In Fig.~\ref{34:fig:3}, we have plotted the capacity of the channel $\mathcal{H}_{L,T}$ as a function of $\varepsilon$ for various values of the input alphabet size $L$ and with both quantization schemes $\mathcal{Q}_{\rm rounding}$ and $\mathcal{Q}(A)$. The threshold probability is $T=10^{-8}$. We observe that the quantization scheme $\mathcal{Q}_{\rm rounding}$ gives the best performance for intermediate-to-large values of $\varepsilon$, while  the quantization scheme $\mathcal{Q}(\mathcal{A})$ performs better when $\varepsilon$ decreases. 
Note that by looking at the optimal input distributions $f_X(x)$ we observe that the shortest runlengths (i.e., the smallest values of $x$) have the highest probabilities. Thus, an error control code for this channel should be designed to  give coded sequences in which small runlengths occur more frequently than longer runlengths. This is the topic of the next section. 
%


\section{Coding Schemes for the Reader-to-Tag Channel}
\label{sec:coding}
Among the  encoding schemes in use in communication standards for RFID applications, one can find several codes that are popular in general communication protocols, such as  NRZ, Manchester, Unipolar RZ, and Miller coding \cite{RFID2010}. Here, we will study the effect of some new encoding schemes for the reader-to-tag channel, considering power constraints (see Section~\ref{sec:power}) and the communication channel described in Section~\ref{sec:channelmodels}, i.e.,   the  Gaussian shift channel. As a reference, we will provide the Manchester code (described in Section~\ref{sec:Manchester}), and two variable-length codes presented in \cite{ros11} (and described in Sections~\ref{sec:10_011} and \ref{sec:101_01101}, respectively) and designed for the bit-shift channel from \cite{ros11,ros09glo}. 


Before describing the specific code constructions, we will briefly explain the concept of constrained coding.

\subsection{Runlength Limitations and Other Coding Constraints}
\label{sec:constrainedCoding}
We may desire and enforce that an encoded sequence satisfies certain constraints specified by a \emph{constraint graph} \cite{Immink91,MSW92,MRS98}. These constraints may, for example, be the power constraints described in Section~\ref{sec:power}, or runlength limitations, or a combination of these constraints. For the purpose of this paper, we shall denote a particular binary runlength limitation as $\mbox{RLL}({\cal L}_0,{\cal L}_1)$, where ${\cal L}_b$ is the set of admissible runlengths of binary symbol $b$. 
In the following, $\mathcal{O}(\cdot)$ refers to the \emph{big O notation} for describing the limiting behavior of functions.

\begin{theorem}
\label{thm:RLL}
If a code satisfying the $\mbox{RLL}([1,L],[1,L])$ limitation, where $[1,L] = \{1,2,\ldots,L\}$, is used on the discretized Gaussian shift channel with quantization scheme $\mathcal{Q}_{\rm rounding}$ and with an ML decoder, then the FER behaves as ${\cal O}(p(L))$ for small $\varepsilon$.
\end{theorem}

\begin{IEEEproof}
By looking at the transition probabilities in (\ref{eq:tran_prob_rounding}), we observe that the dominating error event (in terms of error probability) is when a length-$L$ runlength (the largest allowed by the code) is received as  a length-$(L-1)$ runlength. This is the case since $p(\cdot)$ is an increasing function of its argument. From  (\ref{eq:tran_prob_rounding}), it follows that the FER behaves as ${\cal O}(p(L))$ for small $\varepsilon$, and the result follows.
\end{IEEEproof}


\begin{theorem}
\label{thm:RLLnew}
If a code satisfying the $\mbox{RLL}([1,L],[1,L])$ limitation is used on the discretized Gaussian shift channel with quantization scheme $\mathcal{Q}([1,L])$ and with an ML decoder, then the FER behaves as ${\cal O}(p(L-1/2))$ for small $\varepsilon$.
\end{theorem}

\begin{IEEEproof}
By looking at the transition probabilities in (\ref{eq:tran_prob_q}), we observe that the dominating error event (in terms of error probability) is when a length-$L$ runlength (the largest allowed by the code) is received as a length-$(L-1)$ runlength. Again, this is, as noted in the proof of Theorem~\ref{thm:RLL}, the case since $p(\cdot)$ is an increasing function of its argument. From  (\ref{eq:tran_prob_q}), it follows that the FER behaves as ${\cal O}(p(\alpha(L)))$ for small $\varepsilon$, where (from (\ref{eq:alphabeta}))
\begin{displaymath}
\alpha(L) = \frac{L-1+L}{2(L-(L-1))} = \frac{2L-1}{2} = L-1/2
\end{displaymath}
and the result follows.
%
\end{IEEEproof}

The maximum rate of a constrained code is determined by the \emph{capacity} of the constraint, which can readily be calculated from the constraint graph \cite{Immink91,MSW92,MRS98}. There exist several techniques \cite{Immink91,MSW92,MRS98} for designing an encoder (of code rate upper-bounded by the capacity), and we refer the interested reader to these works for further details.


\subsection{The Manchester Code}
\label{sec:Manchester}

The Manchester code is a very simple block code that maps $0$ into $01$, and $1$ into $10$. The code is popular in many communication protocols, but one can observe that it also satisfies several of the criteria we can derive for a coding scheme to be used on a reader-to-tag discretized Gaussian shift channel: The maximum runlength is two; the average power is $1/2$; the minimum sustainable power is $1/2$; the local minimum power is $1/4$; the minimum Hamming distance is two, and the code is simple to decode. The performance of this code on the discretized Gaussian shift channel will be presented in Section~\ref{sec:sim}.

\subsection{The Code $\{10,011\}$ \cite{ros11,ros09glo}}
\label{sec:10_011}
The variable-length code $\{10,011\}$ is single bit-shift error correcting, i.e., it corrects any single bit-shift on the bit-shift model from \cite{ros11,ros09glo}, 
and has minimum sustainable power $1/2$, local minimum  power $1/3$,   and average power $3/5$. The rate of the code is $2/5$, the minimum runlength is $1$, and the maximum runlength is $3$. The performance of this code on the discretized Gaussian shift channel will be presented in Section~\ref{sec:sim}.

\subsection{The Code $\{101,01101\}$ \cite{ros11}}
\label{sec:101_01101}
The variable-length code $\{101,01101\}$ is single bit-shift error detecting, i.e., it detects any single bit-shift on the bit-shift channel from \cite{ros11,ros09glo}, 
and has minimum sustainable power $3/5$, local minimum  power $1/3$, and  average power $5/8$. The rate of the code is $1/4$, the minimum runlength is $1$, and the maximum runlength is $2$. The performance of this code on the discretized Gaussian shift channel will be presented in Section~\ref{sec:sim}.

\subsection{$\mbox{RLL}(\{1,2\},\{1,2\})$-Limited Codes} \label{sec:newcodes}

The capacity of the constraint $\mbox{RLL}(\{1,2\},\{1,2\})$ is $0.694$. Furthermore, it follows from Theorems~\ref{thm:RLL} and \ref{thm:RLLnew} that, similar to the Manchester code, any code with this runlength limitation has a FER on the order of ${\cal O}(p(2))$ and ${\cal O}(p(3/2))$, for small $\varepsilon$,  on the discretized Gaussian shift channel with quantization schemes $\mathcal{Q}_{\rm rounding}$ and $\mathcal{Q}([1,2])$, respectively.

\begin{example} \label{ex:1}
A two-state, rate-$2/3$ encoder for a $\mbox{RLL}(\{1\},\{1,\ldots,\infty\})$-constrained code is given in \cite{MSW92}. The encoder can be transformed into a four-state encoder for a $\mbox{RLL}(\{1,2\},\{1,2\})$-constrained code by a simple differential mapping. The encoder is shown in Fig.~\ref{fig:(01)a}, while a very simple decoder/demapper is provided in Table~\ref{tab:1}. 
The code has minimum sustainable power $1/3$, local minimum  power $1/5$,  and average power $1/2$.
\end{example}

\begin{figure}[!hbp]
\vspace{-1.0cm}
\centerline{\includegraphics[width=9cm,keepaspectratio=true,angle=0]{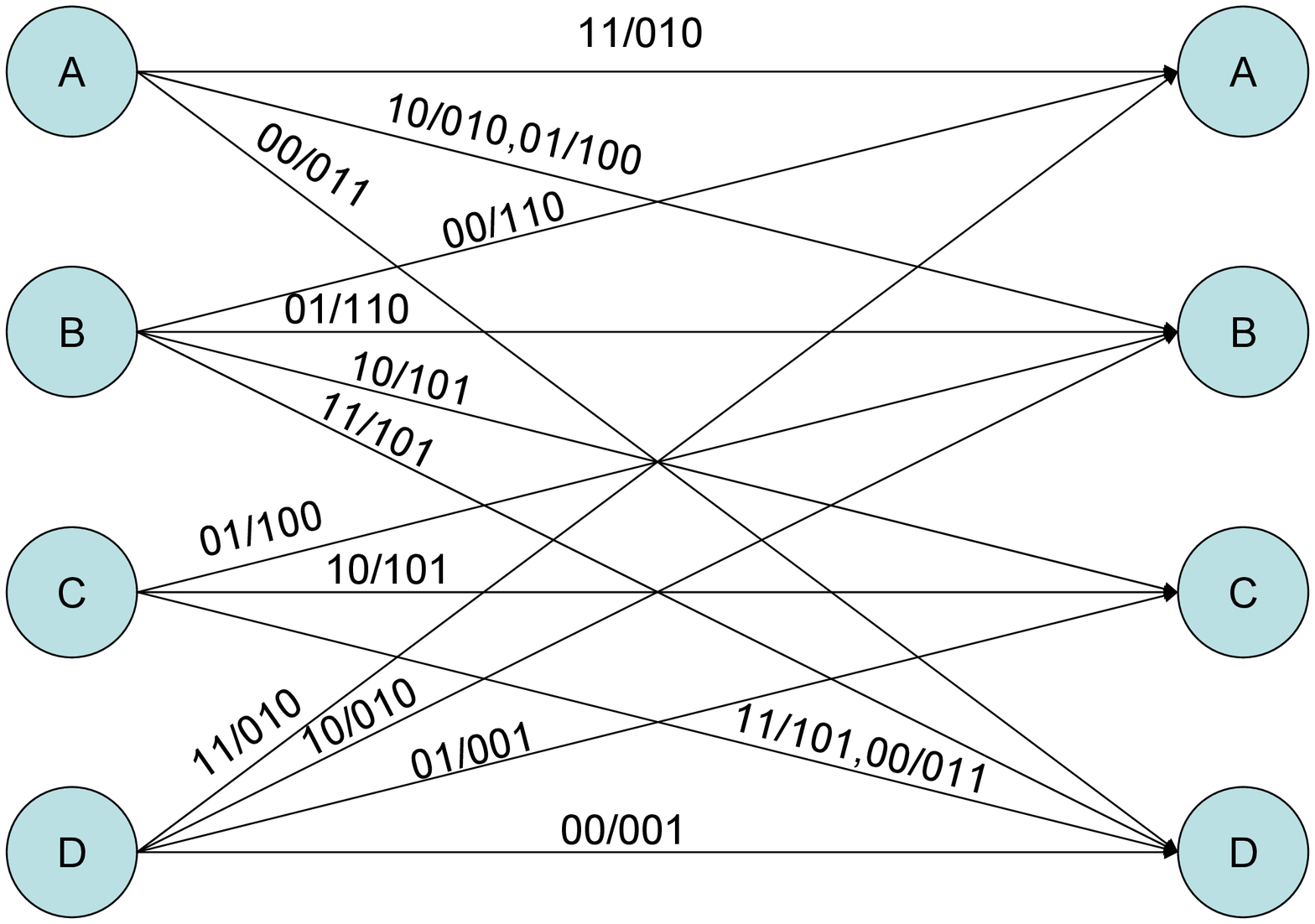}}
\vspace{-0.5cm}
\caption{An encoder for a $\mbox{RLL}(\{1,2\},\{1,2\})$-constrained code.}
\label{fig:(01)a}
\end{figure}

\begin{table}[!t]
\scriptsize \centering \caption{Look-up table decoding of the $\mbox{RLL}(\{1,2\},\{1,2\})$-constrained code from Example~\ref{ex:1} and with the encoder  given in Fig.~\ref{fig:(01)a}. Before decoding, if a run of at least three  zeros or ones is observed, change it to two, since it most likely comes from insertions.}\label{tab:1} 
\def\Hline{\noalign{\hrule height 2\arrayrulewidth}}
\vskip -1.5ex 
\begin{tabular}{c|c|c}
\Hline \\ [-2.0ex]
Current word & Next bits & Decode to \\
\hline%
\\ [-2.0ex] \hline  \\ [-2.0ex]%
%
$000$ \T\B  & Not possible & Detect insertion \\ \hline
$001$ \T & $010$, $001$ & $00$ \\
      \B & $1$, $011$   & $01$ \\ \hline
$010$ \T & $0$, $100$   & $11$ \\
      \B & $110$, $101$ & $10$ \\ \hline
$011$ \T & $0$ & $00$ \\
      \B & ($1$ means insertion) &  \\ \hline
$100$ \T & $1$ & $01$ \\
      \B & ($0$ means insertion) &  \\ \hline
$101$ \T & $010$, $001$ & $11$ \\
      \B & $1$, $011$   & $10$ \\ \hline
$110$ \T & $0$, $100$   & $00$ \\
      \B & $110$, $101$ & $01$ \\ \hline
$111$ \T\B & Not possible & Detect insertion \\
\hline
\end{tabular}
\end{table}


\begin{example}
A code with a very simple encoding and decoding can be obtained by using bit-stuffing. The resulting code is a variable-length code. The encoder keeps the information symbols $u_t$, $t = 1,\ldots,k$, unmodified, but inserts an extra inverted symbol $1-u_t$ if $u_t \equiv t \pmod 2$. The decoder destuffs the extra inserted symbols in a similar way. The encoder has rate $2/3$, minimum sustainable power $1/3$, local minimum  power $1/5$, average power $1/2$, and maximum runlength $2$.
\end{example}

\subsection{$\mbox{RLL}(\{1\},\{1,2\})$-Limited Codes}

The capacity of the constraint $\mbox{RLL}(\{1\},\{1,2\})$ is $0.406$. Thus, a practical rate is no higher than $2/5$. However, the FER on the discretized Gaussian shift channel behaves (for small $\varepsilon$) as ${\cal O}(p(2))$ and ${\cal O}(p(3/2))$  with quantization schemes $\mathcal{Q}_{\rm rounding}$ and $\mathcal{Q}([1,2])$, respectively.  The only advantage over the $\mbox{RLL}(\{1,2\},\{1,2\})$ limitation is a higher power content.


\subsection{$\mbox{RLL}(\{1,3\},\{1,3\})$-Limited Codes} \label{sec:13:13}

The capacity of the constraint  $\mbox{RLL}(\{1,3\},\{1,3\})$ is $0.552$. 

\begin{theorem}
\label{thm:RLL1}
The FER on the discretized Gaussian shift channel with quantization scheme $\mathcal{Q}_{\rm rounding}$ for $\mbox{RLL}(\{1,3\},\{1,3\})$-constrained codes is on the order  of ${\cal O}(p(1))$ for small $\varepsilon$.
\end{theorem}

\begin{IEEEproof}
The decoder works in the following way. Every received  run of length $1$ (on the discretized Gaussian shift channel with quantization scheme $\mathcal{Q}_{\rm rounding}$) is kept as is, and  every received run of length $\geq 2$ is assumed to be a run of length $3$. This decoder makes an error if a run of length $1$ is extended by the Gaussian shift channel to length more than $3/2$ (this happens with probability $p(1)$), or if a run of length $3$ is shortened to less than $3/2$ (this happens with probability $Q \left(\frac{3}{6 \varepsilon} \right) = p(1)$).
\end{IEEEproof}

We remark that on the discretized Gaussian shift channel with quantization scheme $\mathcal{Q}(\{1,3\})$, the error probability is of the same order for small $\varepsilon$, i.e., it behaves as $\mathcal{O}(p(1))$.


\begin{example} \label{ex:2}
A three-state, rate-$1/2$ encoder for a $\mbox{RLL}(\{1,3\},\{1,3\})$-constrained code is depicted in Fig.~\ref{fig:(13)a}, while a very simple decoder/demapper is provided in Table~\ref{tab:2}. The code has minimum sustainable power $1/4$, local minimum  power $1/7$,  and average power $13/24$.
\end{example}

\begin{figure}[!hbp]
\ifonecolumn
\vspace{-2.5cm}
\centerline{\includegraphics[width=15cm,keepaspectratio=true,angle=0]{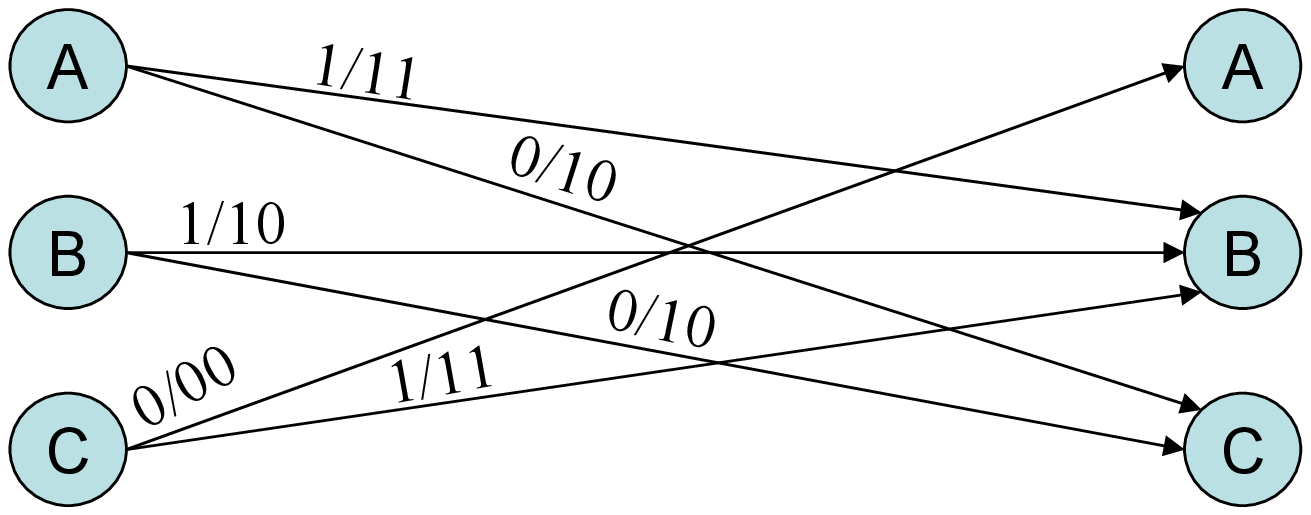}}
\vspace{-5.5cm}
\else
\vspace{-1.5cm}
\centerline{\includegraphics[width=12cm,keepaspectratio=true,angle=0]{Figure7.eps}}
\vspace{-3.0cm}
\fi
\caption{An encoder for a $\mbox{RLL}(\{1,3\},\{1,3\})$-constrained code.}
\label{fig:(13)a}
\end{figure}




\begin{table}[!t]
\scriptsize \centering \caption{Look-up table decoding of the $\mbox{RLL}(\{1,3\},\{1,3\})$-constrained code from Example~\ref{ex:2} and with the encoder given in Fig.~\ref{fig:(13)a}. Before decoding, if a run of two zeros or two ones is observed, change it to three, since it most likely comes from a deletion of a length-three run. Similarly, if a run of four zeros or ones is observed, change it to three.}\label{tab:2}
\def\Hline{\noalign{\hrule height 2\arrayrulewidth}}
\vskip -1.5ex 
\begin{tabular}{c|c|c}
\Hline \\ [-2.0ex]
Current word & Next bit pair & Decode to \\
\hline
\\ [-2.0ex] \hline  \\ [-2.0ex]

$00$ \T\B  & Whatever & $0$ \\ \hline
$01$ \T\B & Not possible & Detect error \\ \hline
$10$ \T & $00$ or $11$ & $0$ \\
     \B & $10$ & $1$ \\ \hline
$11$ \T\B & Whatever & $1$ \\ \hline
\end{tabular}
\end{table}

\begin{example}
A code with a very simple encoding and decoding can be obtained by using bit-stuffing. The resulting code is a variable-length code. The encoder keeps the information symbols $u_t$, $t = 1,\ldots,k$, unmodified, but inserts a pair of extra symbols $(u_t,1-u_t)$ if $u_t \equiv t \pmod 2$. The decoder destuffs the extra inserted symbols in a similar way. The encoder has rate $1/2$, minimum sustainable power $1/4$, local minimum  power $1/7$, average power $1/2$, and allowed runlengths $1$ and $3$.
\end{example}

\subsection{$\mbox{RLL}(\{1\},\{1,3\})$-Limited Codes} \label{sec:1:13}

The capacity of the constraint $\mbox{RLL}(\{1\},\{1,3\})$ is $0.347$. Furthermore, there is no difference in the asymptotic FER (i.e., the FER for small values of $\varepsilon$) with respect to $\mbox{RLL}(\{1,3\},\{1,3\})$-limited codes (the proof of Theorem~\ref{thm:RLL1} holds in this case as well). Thus,  the only advantage over the $\mbox{RLL}(\{1,3\},\{1,3\})$ limitation is a higher power content.

\begin{example} \label{ex:3}
The variable-length $\mbox{RLL}(\{1\},\{1,3\})$-constrained code with codewords $\{01,0111\}$ has rate $1/3$, minimum sustainable power $1/2$, local minimum  power $1/3$,  and average power $2/3$.
\end{example}

\subsection{$\mbox{RLL}(\{1,2,4\},\{1,2,4\})$-Limited Codes}

Codes satisfying the constraints
\ifonecolumn
\begin{align}
\mbox{RLL}(\{1,2,4\},\{1,2,4\}),\; \mbox{RLL}(\{1,2\},\{1,2,4\}),\;\text{and}\; \mbox{RLL}(\{1\},\{1,2,4\}) \notag
\end{align}
\else
\begin{align}
&\mbox{RLL}(\{1,2,4\},\{1,2,4\}),\; \mbox{RLL}(\{1,2\},\{1,2,4\}),\;\text{and} \notag \\
&\mbox{RLL}(\{1\},\{1,2,4\}) \notag
\end{align}
\fi
have capacities $0.811$, $0.758$, and $0.515$, respectively. The latter constraint may be attractive from a power transfer point of view; the two former constraints admit code rates of $4/5$ and $3/4$, respectively, but may be hard to implement. For the $\mbox{RLL}(\{1\},\{1,2,4\})$ constraint, a rate-$1/2$, $6$-state encoder can be designed using the state-splitting algorithm from \cite{adl83}. Finally, we remark that the FER on the discretized Gaussian shift channel  is on the order  of ${\cal O}(p(2))$ and ${\cal O}(p(3/2))$ with the quantization schemes $\mathcal{Q}_{\rm rounding}$ and $\mathcal{Q}(\{1,2,4\})$, respectively, for small values of $\varepsilon$ for these codes.

\subsection{Related Constraints} \label{sec:relatedconstraints}
Any $\mbox{RLL}(\{3^i: i= 0, \ldots,L\},\{3^i: i= 0, \ldots,L\})$-limited code, for any positive integer $L$,  has a  FER of the order of ${\cal O}(p(1))$ (with both quantization schemes $\mathcal{Q}_{\rm rounding}$ and $\mathcal{Q}(\{3^i: i= 0, \ldots,L\})$) for small $\varepsilon$.  This can be shown with a similar argument to that used to prove Theorem~\ref{thm:RLL1}. We remark here that the ${\cal O}(p(1))$ performance guarantee under the quantization scheme $\mathcal{Q}_{\rm rounding}$ assumes that the decoder deals with nonadmissible (with respect to the code) observed runlengths  in the appropriate way. Notice that the capacity seems to approach a limit at about $0.58$ as $L$ increases. Thus, there seems to be no immediate practical advantage on extending these ideas further.

\section{Simulation Results} \label{sec:sim}

\ifonecolumn
\begin{figure}[!t]
\centerline{\includegraphics[width=0.5\columnwidth,keepaspectratio=true,angle=0]{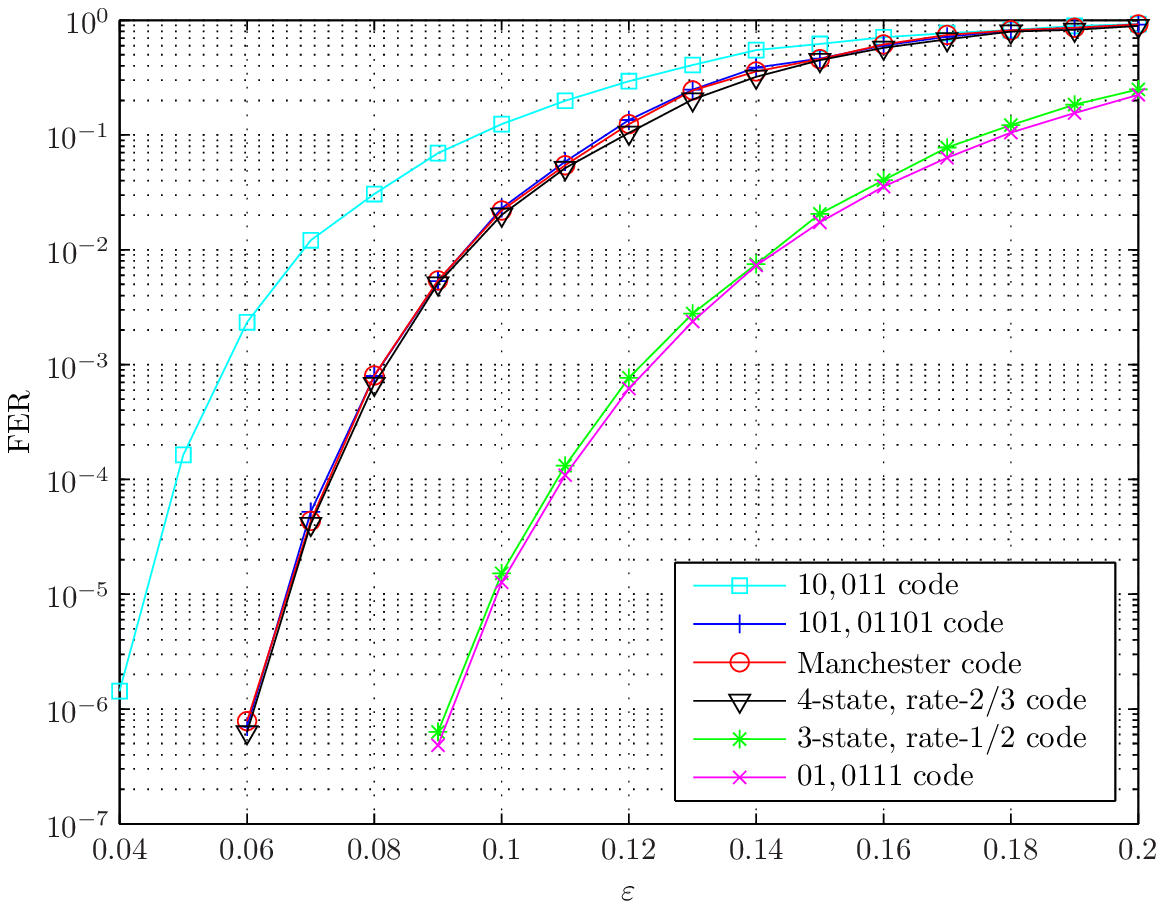}}
\caption{FER on the discretized Gaussian shift channel as a function of $\varepsilon$ for different codes.}
\label{fig:fer}
\end{figure}
\else
\begin{figure}[!t]
\centerline{\includegraphics[width=\columnwidth,keepaspectratio=true,angle=0]{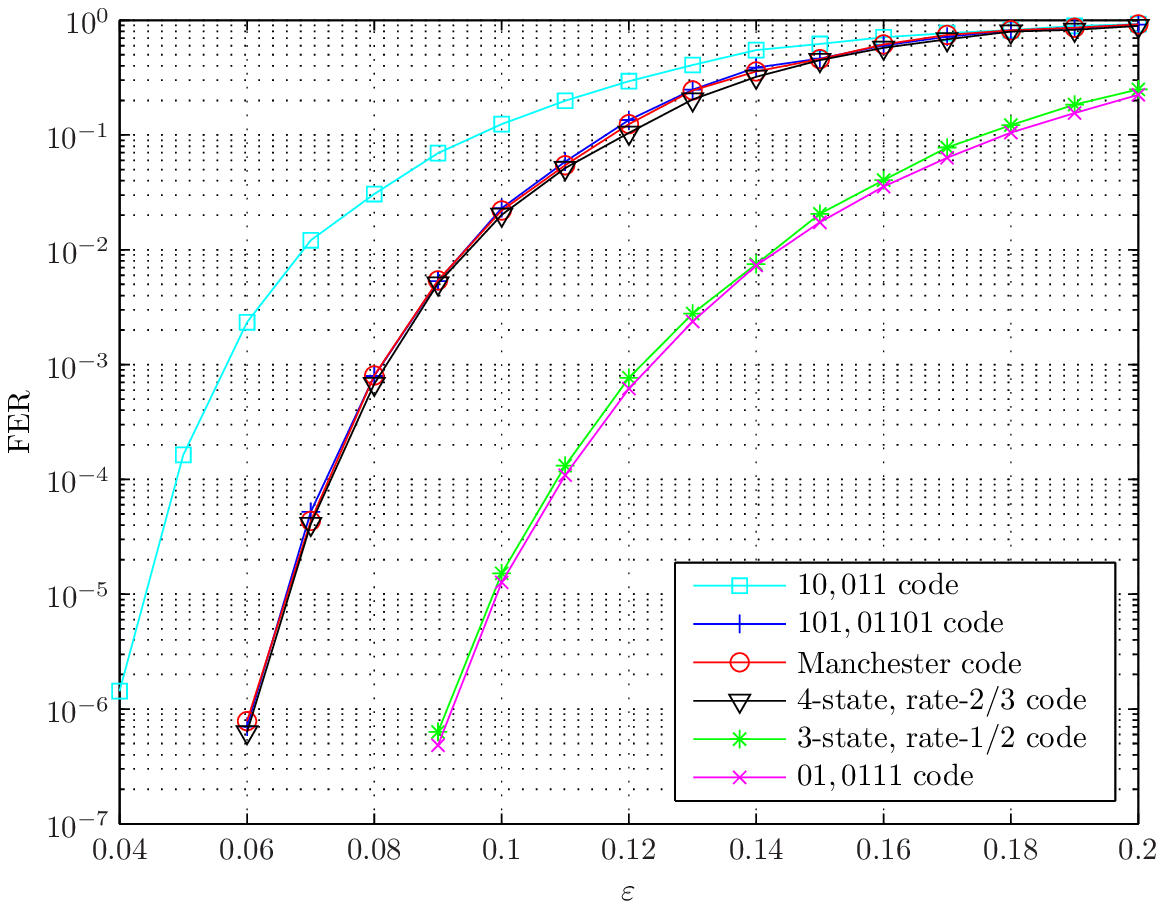}}
\caption{FER on the discretized Gaussian shift channel as a function of $\varepsilon$ for different codes.}
\label{fig:fer}
\end{figure}
\fi

\ifonecolumn
\begin{figure}[!t]
\centerline{\includegraphics[width=0.5\columnwidth,keepaspectratio=true,angle=0]{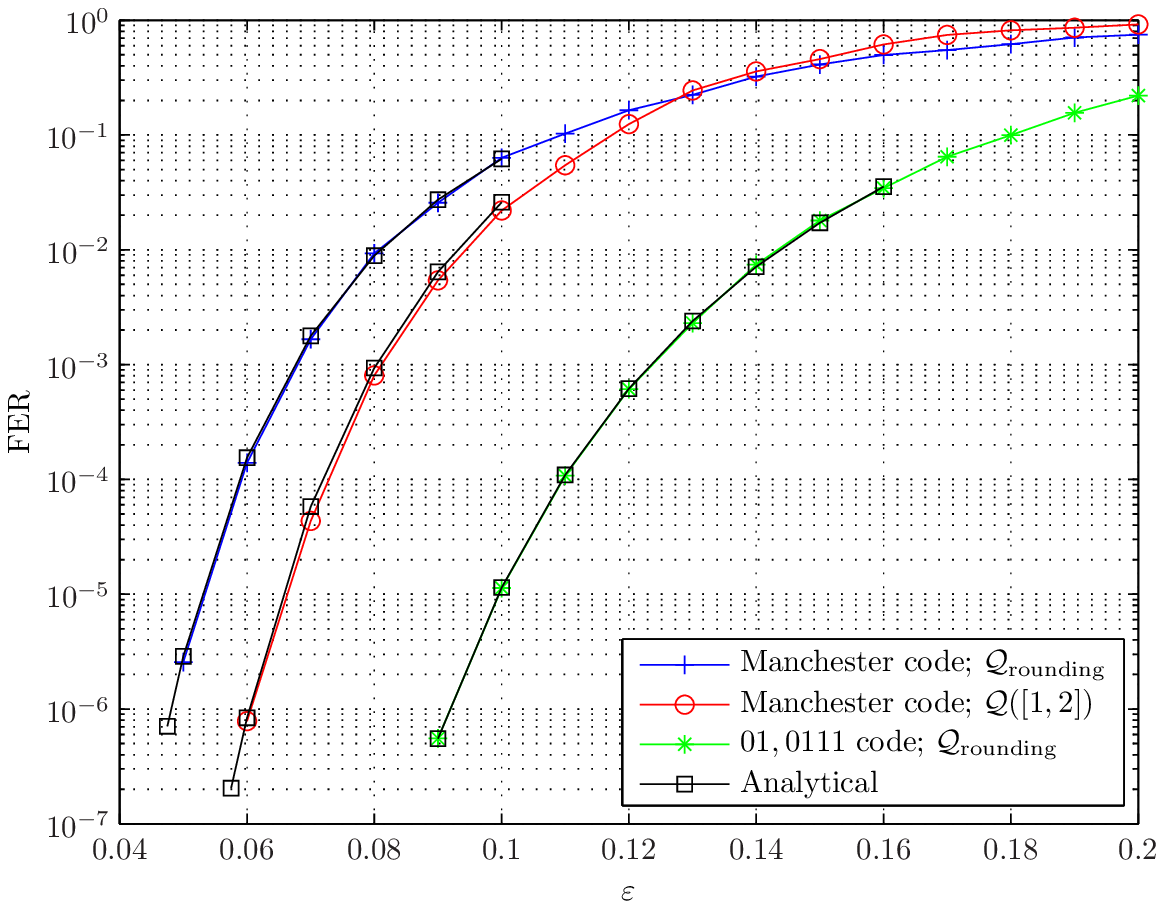}}
\caption{A comparison of the FER (as a function of $\varepsilon$) on the discretized Gaussian shift channel with two different quantization schemes and with  analytical expressions for the asymptotic performance, for the Manchester code and the variable-length code $\{01,0111\}$.}
\label{fig:fer_compare}
\end{figure}
\else
\begin{figure}[!t]
\centerline{\includegraphics[width=\columnwidth,keepaspectratio=true,angle=0]{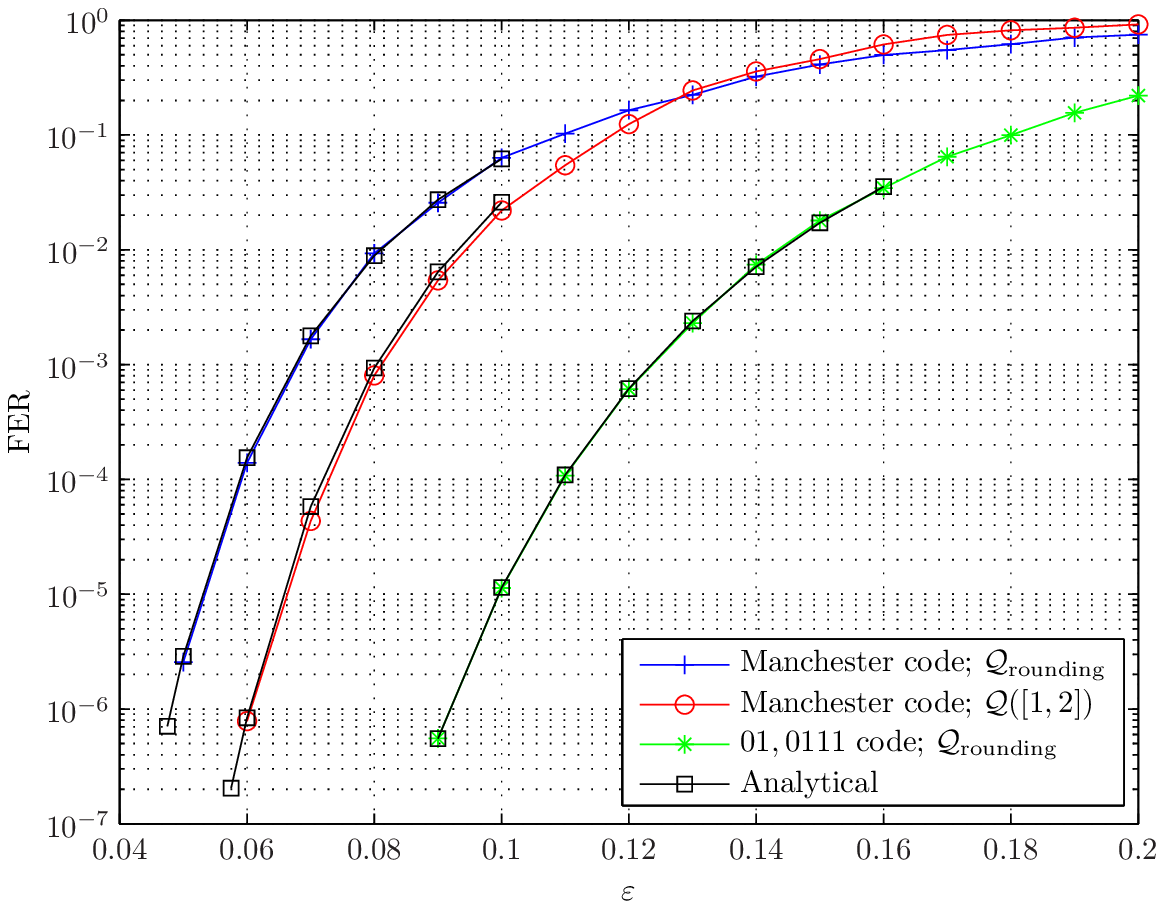}}
\caption{A comparison of the FER (as a function of $\varepsilon$) on the discretized Gaussian shift channel with two different quantization schemes and with  analytical expressions for the asymptotic performance, for the Manchester code and the variable-length code $\{01,0111\}$.}
\label{fig:fer_compare}
\end{figure}
\fi

In this section, we provide some simulation results of some of the above-mentioned codes on the discretized Gaussian shift channel. In particular, we consider the Manchester code from Section~\ref{sec:Manchester}, the $\{10,011\}$ code from Section~\ref{sec:10_011}, and the $\{101,01101\}$ code from Section~\ref{sec:101_01101}, in addition to the newly designed codes from Examples~\ref{ex:1}, \ref{ex:2}, and \ref{ex:3}. For convenience, the information block length $k$  is chosen to be $40$ bits; this allows a simple processing, while the block length is already long enough for the issues addressed in our analytical approach to be valid. The simulation was carried out with software implemented in C++, and the simulation was terminated (for each simulation point) after the transmission of $10^8$ frames or when at least $200$ frame errors were recorded.  This should give sufficient statistical significance, as can also be seen from the figures (all simulation curves are smooth), and hence no error bars are included.

In Fig.~\ref{fig:fer}, we have plotted the FER performance of these codes as function of $\varepsilon$ on the discretized Gaussian shift channel with quantization scheme $\mathcal{Q}([1,2])$ for the Manchester code, for the $\{101,01101\}$ code from Section~\ref{sec:101_01101}, and for the code from  Example~\ref{ex:1}, with quantization scheme  $\mathcal{Q}([1,3])$ for the $\{10,011\}$ code from Section~\ref{sec:10_011}, and with quantization scheme  $\mathcal{Q}(\{1,3\})$ for the remaining codes. As can be observed from the figure, the $\mbox{RLL}(\{1,3\},\{1,3\}$)-constrained code from  Example~\ref{ex:2} and the $\mbox{RLL}(\{1\},\{1,3\}$)-constrained code from Example~\ref{ex:3}  have the best error rate performance, while the variable-length code  $\{10,011\}$ designed in \cite{ros11,ros09glo} for the traditional bit-shift channel has the worst performance among the simulated codes. However, this is not surprising, since this code has not been designed for the discretized Gaussian shift channel.

In Fig.~\ref{fig:fer_compare}, we have compared the performance of two different codes, namely the $\mbox{RLL}(\{1,2\},\{1,2\}$)-constrained Manchester code and the $\mbox{RLL}(\{1\},\{1,3\}$)-constrained code $\{01,0111\}$ from Example~\ref{ex:3} with two different quantization schemes. We have used the quantization schemes simulated in Fig.~\ref{fig:fer} and the quantization scheme $\mathcal{Q}_{\rm rounding}$. Note that the curve for the $\mbox{RLL}(\{1\},\{1,3\}$)-constrained code $\{01,0111\}$ from Example~\ref{ex:3} with quantization scheme $\mathcal{Q}(\{1,3\})$ is not explicitly shown, since there is no noticeable performance difference between the two quantization schemes for this code (as also shown by Proposition~\ref{prop:code01_0111} below). 
On the other hand, there is a significant performance difference for the other code. This is consistent with our earlier discussion in Section~\ref{sec:coding}. In the figure, we also show analytical expressions for the asymptotic performance which depend on both the quantization scheme used and the particular decoding algorithm.  These expressions are given by the propositions below and match perfectly with the simulation results.

\begin{proposition} \label{prop:code01_0111}
For the $\mbox{RLL}(\{1\},\{1,3\})$-constrained code from Example~\ref{ex:3}, the FER (with both quantization schemes $\mathcal{Q}(\{1,3\})$ and $\mathcal{Q}_{\rm rounding}$) on an information block of length $k$ is approximately
\begin{displaymath}
k \cdot p(1) = k \cdot Q\left(\frac{1}{2\varepsilon} \right)
\end{displaymath}
as $\varepsilon$ becomes smaller.
\end{proposition}

\begin{IEEEproof}
The $p(\cdot)$-part of the expression follows from Theorem~\ref{thm:RLL1} (or more precisely, the proof of Theorem~\ref{thm:RLL1}, since the proof holds for both $\mbox{RLL}(\{1\},\{1,3\})$-constrained and $\mbox{RLL}(\{1,3\},\{1,3\})$-constrained codes). The factor $k$ in front follows from the fact that the decoder needs to make exactly one decision for each information symbol.
\end{IEEEproof}

\begin{proposition}
For the Manchester code on an information block of length $k$, the FER is approximately
\begin{displaymath}
\left(3k/2+1/2 \right) \cdot p(3/2) = \left(3k/2+1/2 \right) \cdot Q\left(\frac{1}{3\varepsilon} \right)
\end{displaymath}
and
\begin{displaymath}
k/4 \cdot p(2) = k/4 \cdot Q\left(\frac{1}{4\varepsilon} \right)
\end{displaymath}
for large $k$,
with quantization schemes $\mathcal{Q}([1,2])$ and $\mathcal{Q}_{\rm rounding}$, respectively, as $\varepsilon$ becomes smaller.
\end{proposition}

\begin{IEEEproof}
The $p(\cdot)$-parts of the expressions follow from Theorems~\ref{thm:RLLnew} and \ref{thm:RLL}, respectively. The factor in front (for quantization scheme $\mathcal{Q}([1,2])$) follows from the fact that all runlengths are equally critical, i.e., the probability of a length-$1$ runlength (of zeros or ones) being received as a length-$2$ runlength and vice versa is the same. Thus, the multiplicity in front of the $p(\cdot)$-part will be the average number of runlengths in a codeword, which is exactly $3k/2+1/2$.

With the second quantization scheme ($\mathcal{Q}_{\rm rounding}$), only length-$2$ runlengths are critical, but not all of them, as can be seen from the operation of the decoder. In particular, the decoder works as shown in Table~\ref{tab:Manch}. Observe a window of two consecutive bits of the received sequence (second column of Table~\ref{tab:Manch}). Then, based on the previously decoded bit (first column of Table~\ref{tab:Manch}), decode the observed bits as indicated in the third column of Table~\ref{tab:Manch}, and advance a number of bits (as indicated in the fourth column) for the new window.
\begin{table}[!t]
\scriptsize \centering \caption{Look-up table decoding of the Manchester code.}\label{tab:Manch}
\def\Hline{\noalign{\hrule height 2\arrayrulewidth}}
\vskip -1.5ex 
\begin{tabular}{c|c|c|c}
\Hline \\ [-2.0ex]
Previously decoded bit & Next bit pair & Decode to & Advance \\
\hline
\\ [-2.0ex] \hline  \\ [-2.0ex]

$1$ \T\B  & $10$ & $1$  & $2$  \\ \hline
$1$ \T\B  & $01$ & $0$  & $2$  \\ \hline
$1$ \T\B  & $11$ & $0$  & $1$  \\ \hline
$1$ \T\B  & $00$ & Whatever  & $3$  \\ \hline
$0$ \T\B  & $01$ & $0$  & $2$  \\ \hline
$0$ \T\B  & $10$ & $1$  & $2$  \\ \hline
$0$ \T\B  & $00$ & $1$  & $1$  \\ \hline
$0$ \T\B  & $11$ & Whatever  & $3$  \\ \hline
\end{tabular}
\end{table}
%
%
%

Now, for instance, the sequence $\dots 10.01.01 \dots$ is critical, since it can be received as $\dots 10.10.1 \dots$ (the third bit is deleted), and the decoder from Table~\ref{tab:Manch} is not able to recover it without errors. On the other hand, the sequence $\dots 10.01.10 \dots$ is \emph{not} critical, since the corresponding received sequence $\dots 10.11.0 \dots$ (again the third bit is deleted) is decoded correctly by the decoder of Table~\ref{tab:Manch}. In summary, from Table~\ref{tab:Manch}, the critical sequences are in fact exactly those that have a length-$2$ runlength followed by a length-$1$ runlength. Assuming equally likely transmitted information symbols, with probability $1/4$, a critical pattern occurs in the transmitted sequence for each information symbol, and the expression (valid for large $k$) follows.
\end{IEEEproof}

We have also simulated the case where the communication takes place  on the discretized Gaussian shift channel with simultaneous additive white Gaussian noise (AWGN), and where the Gaussian noise is added (at the bit level) at the output of the discretized Gaussian shift channel. The results, shown in Fig.~\ref{fig:shftandAWGN}, are for the Manchester code with $\mathcal{Q}([1,2])$ 
and hard-decisions as a function of both $\varepsilon$ and the signal-to-noise ratio (SNR), defined as $(A_1-A_0)^2 /  2 R \sigma^2$, where $\sigma$ is the standard deviation  of the AWGN, $R$ is the code rate, and $A_1$ (resp.\ $A_0$) is  the amplitude level of a one (resp.\ zero). For the other simulated codes, we have observed a similarly shaped performance behavior with simultaneous AWGN  (results not included here).  We remark that in the normal mode of RFID reader-to-tag operation, the SNR can be expected to be high.

\ifonecolumn
\begin{figure}[!htp]
\centerline{\includegraphics[width=0.5\columnwidth,keepaspectratio=true,angle=0]{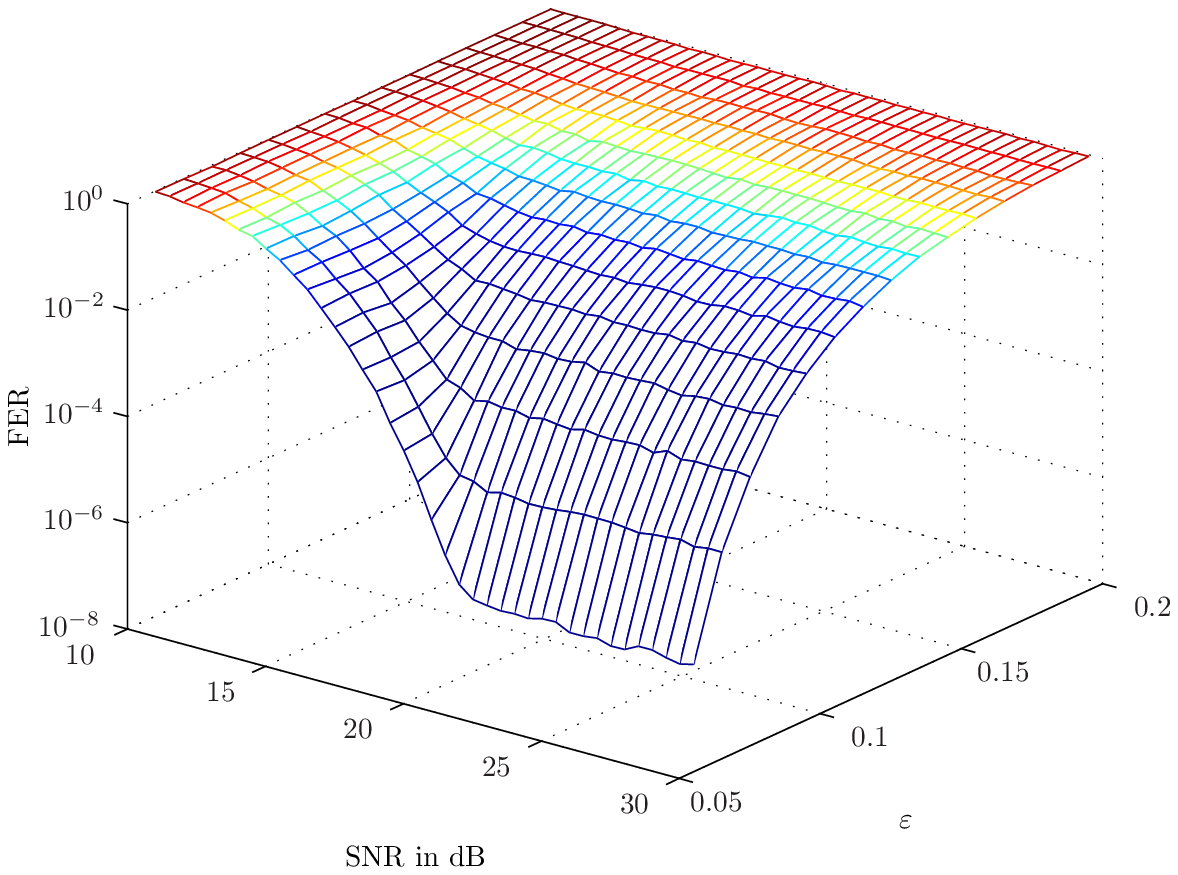}}
\caption{FER on the discretized Gaussian shift channel with simultaneous AWGN as a function of both $\varepsilon$ and the SNR  (in dB) for the Manchester code with $\mathcal{Q}([1,2])$.} 
\label{fig:shftandAWGN}
\end{figure}
\else
\begin{figure}[!htp]
\centerline{\includegraphics[width=\columnwidth,keepaspectratio=true,angle=0]{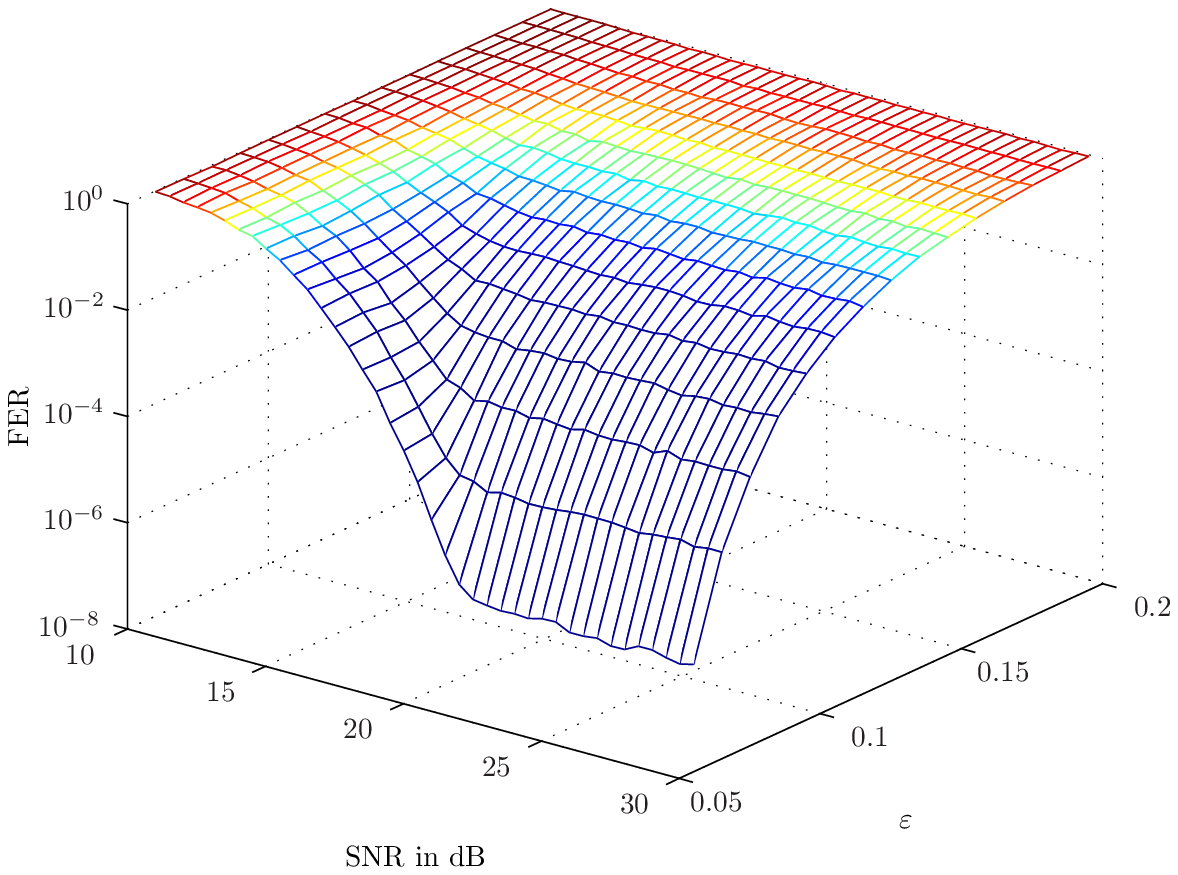}}
\caption{FER on the discretized Gaussian shift channel with simultaneous AWGN as a function of both $\varepsilon$ and the SNR (in dB) for the Manchester code with $\mathcal{Q}([1,2])$.} 
\label{fig:shftandAWGN}
\end{figure}
\fi

Finally, we remark that we have used look-up table decoding in all simulations. For instance, for the codes from  Examples~\ref{ex:1} and \ref{ex:2}, we have used Tables~\ref{tab:1} and \ref{tab:2}, respectively, in the decoding. For the Manchester code, we have used Table~\ref{tab:Manch}.  Further, note that all the codes used in the simulation are local and even a (hard-decision) ML decoder is limited in performance by the issues discussed in our analysis. A soft-decision ML  decoder may improve on this, but will complicate the implementation, something which is undesirable with current technology. A key point of the proposed codes is that they are designed for \emph{error avoidence}, and consequently coding gain is achieved with a very simple decoding procedure.

\section{Conclusion}
\label{sec:conclusion}

In this work, we have discussed a new channel model and code design for near-field passive RFID communication using inductive coupling as a power transfer mechanism. The (discretized) Gaussian shift channel was proposed as a channel model for the reader-to-tag channel when the receiver resynchronizes its internal clock each time a bit is detected. Furthermore, the capacity of this channel was considered, and some new simple codes for error avoidance were presented. Their performance were also compared  to the Manchester code and two previously proposed codes for the bit-shift channel model.

Error avoidance allows a quantification of the coding  gain of a runlength-limited code, and we believe that this quantification adds a new perspective of constrained codes. 

\section*{Acknowledgment}

The authors would like to thank the anonymous reviewers for their valuable comments and suggestions to improve the presentation of the paper.

\balance

\end{document}